\documentclass[journal,12pt,onecolumn, draftclsnofoot]{IEEEtran}

\usepackage[dvips]{graphicx}
\usepackage{amsmath}
\usepackage{amsthm}
\usepackage[flushleft]{threeparttable}
\usepackage{array}
\usepackage{booktabs}
\usepackage{bbm}
\newcommand{\PreserveBackslash}[1]{\let\temp=\\#1\let\\=\temp}
\newcolumntype{C}[1]{>{\PreserveBackslash\centering}p{#1}}
\newcolumntype{R}[1]{>{\PreserveBackslash\raggedleft}p{#1}}
\newcolumntype{L}[1]{>{\PreserveBackslash\raggedright}p{#1}}

\usepackage{bm}
\usepackage{graphicx,subfigure}
\usepackage{epstopdf}
\usepackage{algorithm}
\usepackage{cite}
\usepackage{amsmath}
\usepackage{amssymb}
\usepackage{psfrag}
\usepackage{algpseudocode}
\usepackage{empheq}
\usepackage{latexsym, amsmath, subfigure, color, amsfonts, amssymb,graphicx}
\usepackage{wrapfig}
\usepackage{enumitem}
\usepackage{caption}

\usepackage{bbm, dsfont}
\newtheorem{Theorem}{Theorem}
\newtheorem{Lemma}{Lemma}
\newtheorem{Definition}{Definition}

\newtheorem{Remark}{Remark}
\newtheorem{proposition}{Proposition}
\newtheorem{exmp}{Example}
\IEEEoverridecommandlockouts
\usepackage{tabularx}

\DeclareMathOperator*{\argmax}{\arg\!\max}

\usepackage[T1]{fontenc}
\usepackage{etoolbox}

\makeatletter
\def\hlinewd#1{%
\noalign{\ifnum0=`}\fi\hrule \@height #1 %
\futurelet\reserved@a\@xhline}
\patchcmd{\maketitle}{\@fnsymbol}{\@alph}{}{}  
\makeatother
\def\hlinewd#1{%
\noalign{\ifnum0=`}\fi\hrule \@height #1 %
\futurelet\reserved@a\@xhline}
\patchcmd{\maketitle}{\@fnsymbol}{\@alph}{}{}  
\makeatother

\title{Centralized Caching and Delivery of Correlated Contents over Gaussian Broadcast Channels}

\author{
    \IEEEauthorblockN{Qianqian Yang\IEEEauthorrefmark{1}, Parisa Hassanzadeh\IEEEauthorrefmark{2}, Deniz G\"und\"uz\IEEEauthorrefmark{1}, Elza Erkip\IEEEauthorrefmark{2}
    \thanks{This work has been supported in part by NSF under grant \#1619129, in part by NYU WIRELESS Industrial Affiliates Program, and in part by the European Research Council project BEACON under grant number 677854.}}
    
    \IEEEauthorblockA{%
    \IEEEauthorrefmark{1}Electrical and Electronic Engineering Department, Imperial College London, \{q.yang14, d.gunduz\}@imperial.ac.uk\\\
    \IEEEauthorrefmark{2}Electrical and Computer Engineering Department, New York University, Brooklyn, NY. \{ph990, elza\}@nyu.edu
  }
 
}

\date{}

\begin{document}

\maketitle
\begin{abstract}
Content\makeatletter{\renewcommand*{\@makefnmark}{}\footnotetext{This paper was presented in part at the IEEE Int'l Symp. on Modeling and Opt. in Mobile, Ad Hoc, and Wireless Netw. (WiOpt), Shanghai, China, May 2018 \cite{qianwiopt2018}.}\makeatother} delivery in a multi-user cache-aided broadcast network is studied, where a server holding a database of correlated contents communicates with the users over a Gaussian broadcast channel (BC). The minimum transmission power required to satisfy all possible demand combinations is studied, when the users are equipped with caches of equal size. Assuming uncoded cache placement, a lower bound on the required transmit power as a function of the cache capacity is derived. An achievable centralized caching scheme is proposed, which not only utilizes the user's local caches, but also exploits the correlation among the contents in the database. The performance of the scheme, which provides an upper bound on the required transmit power for a given cache capacity, is characterized. Our results indicate that exploiting the correlations among the contents in a cache-aided Gaussain BC can provide significant energy savings.    
\end{abstract}

\section{Introduction}\label{intro}
Thanks to the decreasing cost and increasing capacity of storage available at mobile devices, \textit{proactive caching} has received significant attention in recent years as a low-cost and effective solution to keep up with the exponentially growing mobile data traffic\cite{samueljsac,MaddahAliCentralized, MohammadQianDenizITW}.
Proactively storing popular contents in cache memories distributed across the network during off-peak traffic periods can greatly reduce both the network congestion and the latency during peak traffic hours. \textit{Coded caching} \cite{MaddahAliCentralized} exploits the broadcast nature of wireless delivery and the contents proactively cached in users' local memories to create multicasting opportunities, even when the users request distinct files, further boosting the benefits of caching. The significant gains of coded caching over traditional uncoded caching schemes have inspired numerous studies, among which \cite{ hassanzadeh2016correlation,qiandeniz2018icc, hassanzadeh2017broadcast, hassanzadeh2017rate, hassanzadeh2016cache, yu2017exact, MohammadDenizerasureTCom, amiri2017gaussian, shirin2016broadcastit, shirin2017broadcastit, Shengerasure2016} are most related to this paper.   

Most of the literature on coded caching considers independent files in the library. However, in many practical settings, files in a cache library can be highly correlated. For example, if we treat chunks of a video file as distinct files to be cached and delivered, these video chunks are typically correlated. Similarly when delivering software updates, each user may request a different version, or updates for a different subset of software packages, which may lead to correlations among requests. In the file correlation model, used in this paper and introduced in \cite{qiandeniz2018icc}, we assume that any subset of the files in the library exclusively share a common part. We present an example of the considered correlation model for three files in Fig. \ref{fig:correlationmodel}, where the common parts of different subsets of files are shown with different colors. This model is fairly general to capture message correlations on the symbol level modeled by arbitrary joint distributions, as more commonly considered in multi-terminal source coding problems\cite{slepian1973noisycoding}, when it is used in conjunction with the  Gray-Wyner network \cite{gray1974source}, which, as described in \cite{hassanzadeh2018rate}, encodes the correlated files into messages with the correlation structure considered in this paper. 

\begin{figure}
     \centering
     \includegraphics[width=0.45\linewidth]{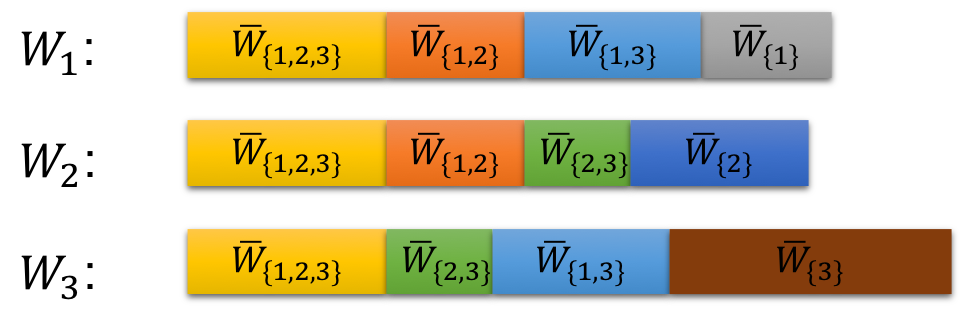}
     \caption{An example of $N=3$ correlated files. Each file consists of $4$ different subfiles with different \textit{commonness} levels.}
     \label{fig:correlationmodel}
 \end{figure}

Delivering correlated contents over an error-free shared link with receiver caches is considered in \cite{hassanzadeh2016correlation, hassanzadeh2017broadcast, hassanzadeh2017rate,qiandeniz2018icc}. In \cite{hassanzadeh2016correlation}, correlations among an arbitrary number of files is exploited by identifying the most representative files, which are then used as references for  compressing the remaining files with respect to the representatives. Correlation among two files is fully exploited in \cite{hassanzadeh2017rate}, in which the files are initially compressed using Gray-Wyner source coding, and an optimal caching scheme is derived for the two-receiver network. This scheme is generalized to more files in \cite{hassanzadeh2017broadcast}, which is optimal for large cache sizes. Arbitrary numbers of users and files are considered in \cite{qiandeniz2018icc}, with the file correlation model illustrated in Fig. \ref{fig:correlationmodel}.

The works in \cite{MohammadDenizerasureTCom, amiri2017gaussian, shirin2016broadcastit, shirin2017broadcastit, Shengerasure2016} consider a more realistic noisy broadcast channel (BC) model from the server to the user. In \cite{shirin2017broadcastit}, the authors consider a degraded BC and a total memory budget, and optimize the cache assignment to the users depending on their channel capacities. A different perspective is taken in \cite{amiri2017gaussian}, which highlights the benefits of caching and coded delivery in terms of the energy-efficiency in a Gaussian BC. However, neither of these papers consider correlation among files.


Following up on \cite{amiri2017gaussian}, in this paper we consider a degraded Gaussian BC model, but rather than independent files, we assume that the files in the library can be arbitrarily correlated as modeled in \cite{qiandeniz2018icc}, and illustrated in Fig. \ref{fig:correlationmodel}. In this model, we have a total of $2^N-1$ subfiles (which can be of size zero), each shared exclusively by a distinct subset of users. We evaluate the performance of this system in terms of the minimum transmission power required to satisfy any demand combination. We derive a lower bound on the transmission power assuming uncoded cache placement, and propose an upper bound, obtained by employing superposition coding and power allocation. For small cache sizes, coded placement and joint encoding scheme is also considered as coded placement is known to better exploit limited cache capacities\cite{ZhiChenXOR} and in asymmetric scenarios \cite{qiandeniz2018tit}. The proposed scheme further exploits the degraded nature of the BC channel by jointly encoding cached contents of the weak users together with the messages targeted at them. This allows the stronger users to receive both the cache contents and the delivered messages of weak users at no additional energy cost. The required transmission power by this scheme meets the derived lower bound that assumes uncoded placement. Through simulations, we show that the proposed correlation-aware joint caching and channel coding scheme reduces the transmission power significantly compared to correlation-ignorant schemes.

The paper is organized as follows. Section~\ref{sys} presents the system model and the problem formulation. A lower bound is presented in Section~\ref{sec:main results}. Two centralized caching and delivery schemes are proposed in Sections~\ref{sec:scheme sep} and \ref{sec:scheme joint} based, respectively, on separate and joint cache-channel coding. Numerical results comparing the proposed upper and lower bounds are provided in Section~\ref{sec:numerical}, and the paper is concluded in Section~\ref{sec:conc}.

\textit{Notations:} The set of integers $\left\{ i, ..., j \right\}$, where $i \le j$, is denoted by $\left[ i:j \right]$, and for $q\in \mathbbm{R}^+$, the set $[1: \lceil q\rceil]$ is denoted shortly by $[q]$. For sets $\mathcal{A}$ and $\mathcal{B}$, we define $\mathcal{A} \backslash \mathcal{B}\triangleq\{x: x \in \mathcal{A}, x\notin \mathcal{B}\}$, and $\left| \mathcal{A} \right|$ denotes the cardinality of $\mathcal{A}$. $\binom{j}{i}$ represents the binomial coefficient if $j\geq i$; otherwise, $\binom{j}{i}=0$. For event $E$, $\mathbbm{1}\{E\}=1$ if $E$ is true; and $\mathbbm{1}\{E\}=0$, otherwise. 


\section{System Model}\label{sys}
Consider a server that holds a database of $N$ correlated files, denoted by $\mathbf{W}=(W_1, ..., W_N)$, each composed of a group of independent subfiles. File $W_i$, $i\in [N]$, consists of $2^{N-1}$ independent subfiles, i.e.,
\small
\begin{equation}
W_i=\{\overline{W}_{\mathcal{S}}: \mathcal{S}\subseteq[N],\;i \in \mathcal{S}\},\notag
\end{equation}
\normalsize
where $\overline{W}_{\mathcal{S}}$ denotes the subfile shared exclusively by the files $\{W_i: i \in \mathcal{S}\}$. For $\mathcal{S}\subseteq [N]$, $|\mathcal{S}|=\ell$, we say that subfile $\overline{W}_{\mathcal{S}}$ has a {\em commonness level of $\ell$}. The subfiles are arranged into $N$ {\em sublibraries}, $L_1,\dots,L_N$, such that $L_\ell$ contains all the subfiles with commonness level of  $\ell$, i.e., 
\small
\begin{equation}
L_\ell=\{\overline{W}_{\mathcal{S}}:\mathcal{S}\subseteq [N],\; |\mathcal{S}|=\ell\}.\notag
\end{equation}
\normalsize
We assume that all the subfiles with the same commonness level, i.e., in the same sublibrary, have the same length, and let subfile $\overline{W}_{\mathcal{S}}\in L_\ell$ be distributed uniformly over the set $[2^{  nR_\ell }]$, where $R_\ell$ is referred to as the rate of subfile $\overline{W}_{\mathcal{S}}$, and $n$ denotes the transmission blocklength, corresponding to $n$ uses of the BC. Let $\mathbf{R} \triangleq (R_1, \ldots, R_N)$. Therefore, all the files are of the same rate of $R$ bits per channel use, given by 
\small
\begin{equation}
R=\sum\limits_{\ell=1}^{N} \binom{N-1}{\ell-1}R_\ell.\notag
\end{equation}
\normalsize
 

Each user is equipped with a cache of size $nM$ bits, where $M$ is called the {\em normalized cache capacity}. Communication takes place in two phases. During the first phase, referred to as the {\em placement phase}, the user caches are filled by the server without the knowledge of user demands. This phase happens during a period of low traffic, and we assume during that phase the channel is noiseless and there are no rate limitations. We consider {\em centralized} caching; that is, the server has the knowledge of the active users in advance, allowing the cache placement to be conducted in a coordinated fashion. At the beginning of the second phase, referred to as the {\em delivery phase}, user $k\in[K]$ requests file $W_{d_k}$ from the library, with $d_k$ uniformly distributed over $[N]$. Let $\mathbf{d}\triangleq (d_1, ..., d_K)$ denote the demand vector. All the requests are satisfied through a Gaussian BC, characterized by a time-invariant channel vector ${\bf h} =(h_1 ,\dots,h_K)$ and additive white Gaussian noise, where $h_k$ denotes the real channel gain between the server and user $k$. The channel gains are fixed, and are known to all the parties. Without loss of generality, we assume $h_1^2 \leq h_2^2 \leq \cdots \leq h_K^2$, such that the users are ordered from the weakest to the strongest. The $i^\text{th}$ channel output at user $k$ is given by
\small
\begin{equation}\label{channel}
Y_{k,i}=h_k\, X_i+\sigma_{k,i}, \notag
\end{equation}
\normalsize
 where $X_i$ and $\sigma_{k,i}\sim \mathcal{N}(0,1)$ denote the channel input and the noise term at user $k$ in the $i^\text{th}$ channel use, respectively, which is independent and identically distributed across time and users. 

For a total transmit power of $P$, an $(n, \mathbf{R}, M, P)$ code for this system consists of:
\begin{itemize}
\item \textbf{$K$ caching functions} $f_{k}$, $k \in [K]$,
\small
\begin{equation}
  f_{k}: [2^{nR}]^N\times \mathbbm{R}^{K} \rightarrow [2^{nMR}],\notag
\end{equation}
\normalsize
such that user $k$'s cache content is given by  $Z_k=f_{k}(\mathbf{W},\mathbf{h})$. Let $\mathbf{Z}\triangleq (Z_1,\dots,Z_K)$.
\item A \textbf{delivery function} $g$, 
\small
\begin{equation}
g: [ 2^{nR}]^N \times [2^{nMR}] \times \mathbbm{R}^{K} \times [N]^K \rightarrow \mathbbm{R}^{n},\notag
\end{equation}
\normalsize
which, for given cache contents $\mathbf Z$, channel gains $\mathbf{h}$, and demand vector $\bf d$, generates the channel input signal,  $X^n(\mathbf{W}, \mathbf Z, \mathbf{d} )=g(\mathbf{W}, \mathbf Z, \mathbf{h}, \mathbf{d})$, transmitted by the server over the Gaussian BC in $n$ channel uses, with $X_i(\mathbf{W}, \mathbf Z, \mathbf{d})$ denoting the $i^\text{th}$ channel input, $i=1, ..., n$. The channel input vector is generated such that its average power over $n$ channel uses is not more than $P$ for any demand vector realization, i.e., 
\small
\begin{equation}
P(\mathbf{W}, \mathbf Z, \mathbf{d}) \triangleq \frac{1}{n} \sum\limits_{i=1}^n X^2_i(\mathbf{W}, \mathbf Z, \mathbf{d})\leq P, \quad \forall\, {\bf d}\in [N]^K.\notag
\end{equation}
\normalsize
\item \textbf{$K$ decoding functions} $\phi_k$, $k\in [K]$,
\begin{equation}
\phi_k: \mathbbm{R}^{n}   \times [2^{nMR}]\times \mathbbm{R}^{K} \times [N]^K \rightarrow [2^{nR}],\notag
\end{equation}
where $\widehat{W}_{d_k}=\phi_k(Y^n(\mathbf{W},\mathbf Z, \mathbf{d}),Z_k,\mathbf{h},\mathbf{d})$, is the reconstruction of $W_{d_k}$ requested by user $k$, and $Y^n(\mathbf{W}, \mathbf Z, \mathbf{d})$ is the channel output at user $k$ for input signal $X^n(\mathbf{W}, \mathbf Z, \mathbf{d})$. 
\end{itemize}

\begin{Definition}
A memory-power pair $(M, P)$ is \textit{achievable} for the system described above, if there exists a sequence of $(n, \mathbf{R}, M, P)$ codes such that 
\small
\begin{equation}
\lim_{n \rightarrow \infty} \mathbbm{P} \Bigg\{\bigcup\limits_{{\bf d}\in [N]^K } \bigcup\limits_{k =1}^K \Big\{\widehat{W}_{d_k}\neq W_{d_k}\Big\}\Bigg\}=0.\notag
\end{equation}
\end{Definition}
\normalsize
For a system with $N$ files and $K$ users, with given channel gains $\mathbf h$, our goal is to characterize the minimum achievable power $P$ as a function of the user cache capacity $M$, i.e.,
\small
\begin{equation}
    P^*(M)\triangleq \inf\{P: (M, P) \mbox{ is achievable}\}.\notag
\end{equation}
\normalsize

\begin{Remark}
In principle different codebooks satisfying different average power constraints can be used for different demand vectors. With the definition above, our goal is to characterize the power constraint that is required to satisfy any demand combination.
\end{Remark}

We conclude this section with the following proposition, which will be frequently referred to in the remainder of the paper.

 \begin{proposition}\label{prop:AWGN}
 \cite{Bergmans1974,bergmans1973degradedBC}
 In a $K$-user degraded Gaussian BC with \small$h_1^2\leq h_2^2\leq\dots\leq h_K^2$,\normalsize distinct messages at rates $\rho_1,\dots,\rho_K$, can be reliably transmitted to users $1,\dots,K$, respectively, iff 
 \small
 \begin{equation}
\rho_k \leq C\Bigg(\frac{h_k^2P_k}{1+h_k^2\sum\limits_{j=k+1}^{K}P_j}\Bigg),~~~  k =1,\dots,K, \label{eq:codebook}
\end{equation}
\normalsize
where $C(x) \triangleq \frac{1}{2} \log_2(1+x)$. This is achieved by superposition coding with Gaussian codewords of power $P_i$, $i=1, ..., K$, to transmit to user $i$. As a consequence, the minimum total transmit power for reliable communication is given by
\small
\begin{equation}
\sum\limits_{k=1}^K P_k \geq \sum\limits_{k=1}^{K}\left(\frac{2^{2 {\rho}_k}-1}{h_k^2}\right) \prod\limits_{j=1}^{k-1}2^{2 {\rho}_j}.\label{eq:power}
\end{equation}
\normalsize
\end{proposition}

\section{Lower Bound}\label{sec:main results}
This section provides a lower bound on the memory-power function, $P^*(M)$ in Theorem ~\ref{lowerbound}, when the placement phase is limited to caching functions that store uncoded contents. We first present a lemma, which will facilitate the proof of Theorem ~\ref{lowerbound}.

We denote by $\mathfrak{D}_d$ the set of all demand combinations such that the first $N_e$ users request distinct files, where $N_e \triangleq \min\{N, K\}$. We note that there are a total of $\binom{N}{N_e}N_e!N^{K-N_e}$ such demand combinations, i.e, $|\mathfrak{D}_d|=\binom{N}{N_e}N_e!N^{K-N_e}$, enumerated as $\mathbf{d}_t\triangleq (d^t_1, ..., d^t_K) \in \mathfrak{D}_d$,\; $t\in [\binom{N}{N_e}\,N_e!\,N^{K-N_e}]$.

\begin{exmp}\label{exmp1}
Consider $N=3$, $K=4$. We have $|\mathfrak{D}_d|=18$ and 
\small
\begin{align}\label{distinctdemand}
&\mathbf{d}_1=\{1, 2, 3, 1\}, 
\mathbf{d}_2=\{1, 2, 3, 2\}, \mathbf{d}_3=\{1, 2, 3, 3\},\nonumber\\
&\mathbf{d}_4=\{1, 3, 2, 1\}, 
\mathbf{d}_5=\{1, 3, 2, 2\}, \mathbf{d}_6=\{1, 3, 2, 3\},\nonumber\\
&\mathbf{d}_7=\{2, 3, 1, 1\}, 
\mathbf{d}_8=\{2, 3, 1, 2\}, \mathbf{d}_9=\{2, 3, 1, 3\},\nonumber\\
&\mathbf{d}_{10}=\{3, 2, 1, 1\}, 
\mathbf{d}_{11}=\{3, 2, 1, 2\}, \mathbf{d}_{12}=\{3, 2, 1, 3\},\nonumber\\
&\mathbf{d}_{13}=\{3, 1, 2, 1\}, 
\mathbf{d}_{14}=\{3, 1, 2, 2\}, \mathbf{d}_{15}=\{3, 1, 2, 3\},\nonumber\\
&\mathbf{d}_{16}=\{2, 1, 3, 1\}, 
\mathbf{d}_{17}=\{2, 1, 3, 2\}, \mathbf{d}_{18}=\{2, 1, 3, 3\}.
\end{align}
\normalsize
\end{exmp}

\begin{Lemma}
There exist random variables $X_{\mathbf{d}_t},$ $Y_{1, \mathbf{d}_t},$ $...,$ $ Y_{N_e, \mathbf{d}_t}$, where for $X_{\mathbf{d}_t}=x$, $x\in \mathbbm{R}$,  
\begin{equation}
Y_{k, \mathbf{d}_t}|x\sim N(h_kx, 1),~~k\in [N_e] \notag \end{equation}
and random variables $U_{1, \mathbf{d}_t},$ $...,$ $ U_{N_e-1, \mathbf{d}_t}$, such that  
\begin{equation}
U_{1,\mathbf{d}_t}-\cdots-U_{N_e-1, \mathbf{d}_t}-X_{\mathbf{d}_t}-Y_{N_e, \mathbf{d}_t}-\cdots-Y_{1, \mathbf{d}_t}\notag
\end{equation}
forms a Markov chain, and 
\small
\begin{align}\label{lemma1}
&H(W_{d^t_1})+\epsilon_n\leq \frac{1}{n} I(W_{d^t_1}; Z_1)+I(U_{1, \mathbf{d}_t}; Y_{1, \mathbf{d}_t});\notag\\
&H(W_{d^t_k}|W_{d^t_{k-1}}, ..., W_{d^t_1})+\epsilon_n\leq \frac{1}{n} I(W_{d^t_k}; Z_1, ..., Z_k|W_{d^t_{k-1}}, ..., W_{d^t_1})+I(U_{k,\mathbf{d}_t}; Y_{k,\mathbf{d}_t}|U_{k-1,\mathbf{d}_t}), k\in [2:N_e-1];\notag\\
&H(W_{d^t_{N_e}}|W_{d^t_{N_e-1}}, ..., W_{d^t_1})+\epsilon_n\leq \frac{1}{n} I(W_{d^t_{N_e}}; Z_1, ..., Z_{N_e}|W_{d^t_{N_e-1}}, ..., W_{d^t_1})+I(X_{\mathbf{d}_t}; Y_{N_e, \mathbf{d}_t}|U_{N_e-1, \mathbf{d}_t}),
\end{align}
\normalsize
where $\epsilon_n$ goes to zero as $n \rightarrow \infty$.
\begin{proof}
The proof is similar to the proof of \cite[Lemma 14]{shirin2017broadcastit}, which we omit here. 
\end{proof}
\end{Lemma}
\begin{Theorem}\label{lowerbound}
For the caching problem described in Section~\ref{sys} with uncoded cache placement phase, the optimal memory-power function, $P^*(M)$, is lower bounded as
\small
\begin{align}
P^*(M)\geq P_{LB}(M)\triangleq\sum\limits_{k=1}^{\min\{N, K\}}\left(\frac{2^{2\tilde{\rho}_k}-1}{h_k^2}\right) \prod\limits_{j=1}^{k-1}2^{2\tilde{\rho}_j},\label{pm:2}
\end{align}
\normalsize
\small
\begin{align}\label{rstar}
\tilde{\rho}_k\triangleq \max\left\{\sum\limits_{\ell=0}^{N-k}\binom{N-k}{\ell}R_{\ell+1}-M, \;0\right\},  \forall\, k \in[K].
\end{align}
\normalsize
\begin{proof}

For any demand vector $\mathbf{d}_t\triangleq (d^t_1, ..., d^t_K) \in \mathfrak{D}_d$, we have $H(W_{d^t_1})=R$, and 
\small
\begin{subequations}
\begin{align}
H(W_{d^t_k}|W_{d^t_{k-1}}, ..., W_{d^t_1})&=H\left(\bigcup\limits_{\substack{\mathcal{S}\subseteq [N]\\  \mathcal{S}\ni d^t_k}}\overline{W}_{\mathcal{S}}|\bigcup\limits_{\substack{\mathcal{S}\subset [N]\\ \{d^t_{k-1}, ..., d^t_1\}\cap\mathcal{S}\neq \emptyset}}\overline{W}_{\mathcal{S}}\right)\label{lemma1:a}\\
&=H\left(\bigcup\limits_{\substack{\mathcal{S}\subseteq [N]\setminus\{d^t_{k-1}, ..., d^t_1\}\\\mathcal{S}\ni d^t_k}}\overline{W}_{\mathcal{S}}\right)=\sum\limits_{\ell=0}^{N-k}\binom{N-k}{\ell}R_\ell, \qquad k \in[2, ..., N_e],\label{lemma1:c}
\end{align}
\end{subequations}
\normalsize
where \eqref{lemma1:a} follows from the fact that $W_i=\bigcup_{\substack{\mathcal{S}\subseteq [N]\\ \mathcal{S} \ni i}}\overline{W}_{\mathcal{S}}$, $\forall i\in [N]$, and  \eqref{lemma1:c} follows due to the independence of the subfiles. Similarly, we have \small $I(W_{d^l_1}; Z_1)=I\left(\bigcup\nolimits_{\substack{\mathcal{S}\subseteq [N]\\ \mathcal{S} \ni d^1_k }}\overline{W}_{\mathcal{S}}; Z_1\right)$,\normalsize and 
\small
\begin{subequations}
\begin{align}
I(W_{d^t_k}; Z_1, ..., Z_k|W_{d^t_{k-1}}, ..., W_{d^t_1})
&=I\left(\bigcup\limits_{\substack{\mathcal{S}\subseteq [N]\\ \mathcal{S}\ni d^t_k}}\overline{W}_{\mathcal{S}}; Z_1, ..., Z_k\bigg|\bigcup\limits_{\substack{\mathcal{S}\subseteq [N]\\ \{d^t_{k-1}, ..., d^t_1\}\cap\mathcal{S}\neq \emptyset}}\overline{W}_{\mathcal{S}}\right)\label{lemma1:d}\\
&\leq I\left(\bigcup\limits_{\substack{\mathcal{S}\subseteq [N]\setminus\{d^t_{k-1}, ..., d^t_1\}\\\mathcal{S}\ni d^t_k}}\overline{W}_{\mathcal{S}}; Z_1, ..., Z_k\right),\label{lemma1:e}
\end{align}
\end{subequations}
\normalsize
 for $k\in [2: N_e]$, where \eqref{lemma1:e} follows due to the independence of the subfiles and uncoded cache placement. Thus, for $n$ sufficiently large, we can rewrite \eqref{lemma1} as 
\small
\begin{align}\label{lemma11}
&R\leq \frac{1}{n} I\left(\bigcup\limits_{\substack{\mathcal{S}\subseteq [N]\\ \mathcal{S}\ni d^1_k}}\overline{W}_{\mathcal{S}}; Z_1\right)+I\left(U_{1, \mathbf{d}_t}; Y_{1, \mathbf{d}_t}\right);\notag\\
&\sum\limits_{\ell=0}^{N-k}\binom{N-k}{l}R_{\ell+1}\leq \frac{1}{n} I\left(\bigcup\limits_{\substack{\mathcal{S}\subseteq [N]\setminus\{d^t_{k-1}, ..., d^t_1\}\\ \mathcal{S}\ni d^t_k }}\overline{W}_{\mathcal{S}}; Z_1, ..., Z_k\right)+I(U_{k, \mathbf{d}_t}; Y_{k, \mathbf{d}_t}|U_{k-1, \mathbf{d}_t}), k\in [2:N_e-1];\notag\\
&\sum\limits_{l=0}^{N-N_e}\binom{N-N_e}{\ell}R_{\ell+1}\leq \frac{1}{n} I\left(\bigcup\limits_{\substack{\mathcal{S}\subseteq [N]\setminus\{d^t_{N_e-1}, ..., d^t_1\}\\\mathcal{S}\ni d^t_{N_e}}}\overline{W}_{\mathcal{S}}; Z_1, ..., Z_{N_e}\right)+I(X_{\mathbf{d}_t}; Y_{N_e, \mathbf{d}_t}|U_{N_e-1, \mathbf{d}_t}).
\end{align}
\normalsize
For degraded Gaussian BC described in Section~\ref{sys}, we have \cite{bergmans1973degradedBC}
\small
\begin{align}
I(U_{k, \mathbf{d}_t}; Y_{k, \mathbf{d}_t}|U_{k-1, \mathbf{d}_t})\leq \frac{1}{2}\log_{2}\left(1+\frac{h_k^2P_k(\mathbf{d}_t)}{h_k^2\sum\limits_{j=k+1}^{N_e}P_j(\mathbf{d}_t)+1}\right),\notag
\end{align}
\normalsize
for $k=1, ..., N_e$, where we set $U_{0, \mathbf{d}_t}\triangleq 0$ and $U_{N_e, \mathbf{d}_t}\triangleq X_{\mathbf{d}_t}$. Thus, with \eqref{lemma11} and according to Proposition \ref{prop:AWGN}, the required average transmission power to satisfy any demand vector $\mathbf{d}_t \in \mathfrak{D}_d$ is lower bounded by 
\small
\begin{align}
P(\mathbf{d}_t) \geq \sum\limits_{k=1}^{N_e}P_k(\mathbf{d}_t)=q(c_1(\mathbf{d}_t), ..., c_{N_e}(\mathbf{d}_t)),  \notag
\end{align}
\normalsize
where 
\small
\begin{subequations}\label{lowerP}
\begin{align}
&q(c_1(\mathbf{d}_t), ..., c_{N_e}(\mathbf{d}_t))\triangleq \sum\limits_{k=1}^{N_e}\left(\frac{2^{2c_k(\mathbf{d}_t)-1}}{h_k^2}\right) \prod\limits_{j=1}^{k-1}2^{2c_j(\mathbf{d}_t)};\label{defi:q}\\
&c_k(\mathbf{d}_t)\triangleq \sum\limits_{\ell=0}^{N-k}\binom{N-k}{\ell}R_{\ell+1}-I\left(\bigcup\limits_{\substack{\mathcal{S}\subseteq [N]\setminus\{d^t_{k-1}, ..., d^t_1\}\\  \mathcal{S}\ni d^t_k }}\overline{W}_{\mathcal{S}}; Z_1, ..., Z_k\right), k\in[N_e].\label{defi:ck}
\end{align}
\end{subequations}
\normalsize
It is proved in \cite[Appendix B]{amiri2017gaussian} that $q(\cdot)$ is a convex function of $\left(C_1(\mathbf{d}_t), ..., C_{N_e}(\mathbf{d}_t)\right)$. Thus, the optimal achievable power is lower bounded by
\small
\begin{subequations}
\begin{align}
P^*(M)\geq &\frac{1}{|\mathfrak{D}_d|} \sum\limits_{t=1}^{|\mathfrak{D}_d|} P(\mathbf{d}_t)
\geq\frac{1}{|\mathfrak{D}_d|}\sum\limits_{t=1}^{|\mathfrak{D}_d|} q(c_1(\mathbf{d}_t), ..., c_{N_e}(\mathbf{d}_t))
\\\geq &q\left(\frac{1}{|\mathfrak{D}_d|}\sum\limits_{t=1}^{|\mathfrak{D}_d|}c_1(\mathbf{d}_t), ..., \frac{1}{|\mathfrak{D}_d|}\sum\limits_{t=1}^{|\mathfrak{D}_d|}c_{N_e}(\mathbf{d}_t)\right)
\geq\sum\limits_{k=1}^{N_e}\left(\frac{2^{2\tilde{\rho}_k}-1}{h_k^2}\right) \prod\limits_{j=1}^{k-1}2^{2\tilde{\rho}_j},\label{pm:1}
\end{align}
\end{subequations}
\normalsize
where we recall that  \small\[\tilde{\rho}_k\triangleq \max\left\{\sum\limits_{\ell=0}^{N-k}\binom{N-k}{\ell}R_{\ell+1}-M, 0\right\}. \]\normalsize\eqref{pm:1} follows from the convexity of $q(\cdot)$, and \eqref{pm:1} holds since $\frac{1}{|\mathfrak{D}_d|}\sum\limits_{t=1}^{|\mathfrak{D}_d|}c_k(\mathbf{d}_t)\geq \tilde{\rho}_k$, $\forall k \in [N_e]$, which we will prove in the following. For any $k \in [N_e]$, we divide all the demands $\mathbf{d}_t \in \mathfrak{D}_d$ into $|\mathfrak{D}_d|/k$ disjoint groups, where each group has $k$ demand vectors such that $d^{t_1}_k \in \{d^{t_2}_1, ..., d^{t_2}_{k-1}\}$, and $d^{t_2}_k \in \{d^{t_1}_1, ..., d^{t_1}_{k-1}\}$, if demand vectors $\mathbf{d}_{t_1}$ and $\mathbf{d}_{t_2}$ are in the same group, and $t_1\neq t_2$.

In Example \ref{exmp1}, there are $18$ demand vectors in $\mathfrak{D}_d$ listed in \eqref{distinctdemand}. For $k=3$, one partition that meets the above condition is 
\small
\begin{align}
&G_1=\{\mathbf{d}_1, \mathbf{d}_4, \mathbf{d}_7\}, G_2=\{\mathbf{d}_2, \mathbf{d}_5, \mathbf{d}_8\}, G_3=\{\mathbf{d}_3, \mathbf{d}_6, \mathbf{d}_9\}, \nonumber\\
&G_4=\{\mathbf{d}_{10}, \mathbf{d}_{13}, \mathbf{d}_{16}\}, G_5=\{\mathbf{d}_{11}, \mathbf{d}_{14}, \mathbf{d}_{17}\}, G_6=\{\mathbf{d}_{12}, \mathbf{d}_{15}, \mathbf{d}_{18}\},\notag
\end{align}
\normalsize
where $G_j$, $j\in [6]$, denotes one group that satisfies $d^{t_1}_3 \in \{d^{t_2}_1, d^{t_2}_{2}\}$, and $d^{t_2}_3 \in \{d^{t_1}_1, d^{t_1}_{2}\}$,   $\forall \mathbf{d}_{t_1}, \mathbf{d}_{t_2} \in G_j, t_1\neq t_2$.  

We denote the index of the $s^\text{th}$ demand vector in the $j^\text{th}$ group by $t_{js}$. Thus, 
\small
\begin{subequations}
\begin{align}
\frac{1}{|\mathfrak{D}_d|} \sum\limits_{t=1}^{|\mathfrak{D}_d|}I&\left(\bigcup\limits_{\substack{\mathcal{S}\subseteq [N]\setminus\{d^t_{k-1}, ..., d^t_1\}\\ \mathcal{S}\ni d^t_k }}\overline{W}_{\mathcal{S}}; Z_1, ..., Z_k\right)= \frac{1}{|\mathfrak{D}_d|} \sum\limits_{j=1}^{|\mathfrak{D}_d|/k}\sum\limits_{s=1}^{k}I\left(\bigcup\limits_{\substack{\mathcal{S}\subseteq [N]\setminus\{d^{t_{js}}_{k-1}, ..., d^{t_{js}}_1\}\\ \mathcal{S}\ni d^{t_{js}}_k  }}\overline{W}_{\mathcal{S}}; Z_1, ..., Z_k\right)\label{mutual:1}\\
&=\frac{1}{|\mathfrak{D}_d|} \sum\limits_{j=1}^{|\mathfrak{D}_d|/k}I\left(\bigcup\limits_{s\in[k]}\bigcup\limits_{\substack{\mathcal{S}\subseteq [N]\setminus\{d^{t_{js}}_{k-1}, ..., d^{t_{js}}_1\}\\ \mathcal{S}\ni d^{t_{js}}_k }}\overline{W}_{\mathcal{S}}; Z_1, ..., Z_k\right)\label{mutual:2}\\
&\leq \frac{1}{|\mathfrak{D}_d|} \sum\limits_{j=1}^{|\mathfrak{D}_d|/k}\min\left\{H\left(\bigcup\limits_{s\in[k]}\bigcup\limits_{\substack{\mathcal{S}\subset [N]\setminus\{d^{t_{js}}_{k-1}, ..., d^{t_{js}}_1\}\\ d^{t_{js}}_k \in \mathcal{S}}}\overline{W}_{\mathcal{S}}\right), H( Z_1, ..., Z_k) \right\}\label{mutual:3}\\
&=\min\left\{\sum\limits_{\ell=0}^{N-k}\binom{N-k}{\ell}R_{\ell+1}, M \right\},\label{mutual:4}
\end{align}
\end{subequations}
\normalsize
where \eqref{mutual:1} is derived by writing the summation with regards to the groups; \eqref{mutual:2} follows the independence of subfiles and the fact that  
\small
\[\bigcap\limits_{s\in[k]}\left(\bigcup\limits_{\substack{\mathcal{S}\subset [N]\setminus\{d^{t_{js}}_{k-1}, ..., d^{t_{js}}_1\}\\ d^{t_{js}}_k \in \mathcal{S}}}\overline{W}_{\mathcal{S}}\right)=\emptyset,\]
\normalsize
since $d^{t_{js_1}}_k \in \{d^{t_{js_2}}_1, ..., d^{t_{js_2}}_{k-1}\}$, while $d^{t_{js_2}}_k \in \{d^{t_{js_1}}_1, ..., d^{t_{js_1}}_{k-1}\}$, if $s_1 \neq s_2$, $\forall s_1, s_2 \in [k], j\in [|\mathfrak{D}_d|/k]$; \eqref{mutual:3} follows since mutual information is no larger than the entropy of each component. \eqref{mutual:4} follows from the size of the subfiles and the cache capacity. Substituting \eqref{defi:ck} and \eqref{mutual:4} into \eqref{pm:1}, we have proven \eqref{pm:2}. Thus, the proof of Theorem~\ref{lowerbound} is completed. 
\end{proof}
\end{Theorem}


\section{Cache-Aided Superposition Coding} \label{sec:scheme sep}
We propose a centralized caching and delivery scheme, which  employs  superposition  coding to deliver  coded  messages  over  the  Gaussian  BC \cite{bergmans1973degradedBC,Bergmans1974}, where  the  coded  messages are generated  taking into account the correlation among the requested files as well as the channel gains. As in \cite{qiandeniz2018icc,hassanzadeh2017rate,hassanzadeh2017broadcast}, the scheme operates by treating the sublibraries independently during the placement and delivery phases to determine the cache content and messages  targeted at each user, which are then jointly delivered over the BC. For clarity, the scheme is first explained on a simple example. 

\textbf{Example 2.} Consider $K=3$ users with channel gains $h_1^2\leq h_2^2\leq h_3^2$, and a database of $N=3$ files as in Fig.~\ref{fig:correlationmodel} with sublibraries: 
\begin{itemize}
\item $L_1 = \{\overline{W}_{\{1\}}, \overline{W}_{\{2\}},\overline{W}_{\{3\}}   \}$, each with rate $R_1$.
\item $L_2 = \{\overline{W}_{\{1,2\}}, \overline{W}_{\{2,3\}},\overline{W}_{\{1,3\}} \} $, each with rate $R_2$.
\item $L_3 = \{\overline{W}_{\{1,2,3\}}\}$, with rate $R_3$.
\end{itemize}
Assume that each user has a normalized cache  capacity of $M = R_1+R_2+\frac{1}{3}R_3$. 

$\circ$ \textbf{Placement Phase:}
Placement is carried out independently across sublibraries. Assume that each user divides its cache into three portions with normalized capacities $R_1$, $R_2$, and $\frac{1}{3}R_3$, allocated for files from sublibraries $L_1$, $L_2$ and $L_3$, respectively. We remark that this cache capacity allocation is not optimal, and the proposed scheme further optimizes the allocation as described in Sec.~\ref{sec:scheme general}. 
We use the prefetching policy proposed in~\cite{yu2017exact}, which divides the subfiles in sublibrary $L_\ell$ into three non-overlapping parts, each of size $\frac{1}{3}nR_\ell$ bits. Then, user $k$ caches 
\begin{align}
Z_k=\Big\{\overline{W}_{\{1\}, \{k\}},& \overline{W}_{\{2\}, \{k\}}, \overline{W}_{\{3\}, \{k\}}, \overline{W}_{\{1, 2\}, \{k\}}, \overline{W}_{\{2, 3\}, \{k\}}, \overline{W}_{\{1, 3\}, \{k\}}, \overline{W}_{\{1, 2, 3\}, \{k\}}\Big\}, \notag
\end{align}
where $\overline{W}_{\mathcal{S}, \{k\}}$ denotes the $k^\text{th}$ part of subfile $\overline{W}_{\mathcal{S}}$ cached at user $k\in [3]$.
 
 $\circ$ \textbf{Delivery Phase:}
Once the demand vector is revealed, the server computes the messages intended for each user, independently for each sublibrary, and delivers them over the BC via superposition coding with Gaussian codewords. The total transmit power is given in Proposition \ref{prop:AWGN}, which depends on the rate of messages intended for each user. Consider the demand vector $\mathbf{d}=(1, 2, 3)$. User 1, the weakest user, needs subfiles $\{\overline{W}_{\{1\}},\overline{W}_{\{1,2\}},\overline{W}_{\{1,3\}},\overline{W}_{\{1,2,3\}}\}$ to reconstruct $W_1$. User 2 requires the four subfiles corresponding to file $W_2$, but having a better channel than user 1. It can also decode the messages targeted at user 1. Similarly, user 3 can decode the messages indented for both of the weaker users. User messages from each sublibrary are determined as follows.
 
 \begin{itemize}
 \item Sublibrary $L_1$: Based on the demand vector, all subfiles in $L_1$ are required by the users. User 1 needs to receive $\overline{W}_{\{1\}, \{2\}}$ and $\overline{W}_{\{1\}, \{3\}}$, whose targeted message, denoted by $V_{1,\mathbf{d}}(L_1)$, is generated as follows:
 \begin{align}
 V_{1,\mathbf{d}}(L_1)= \left\{
 \overline{W}_{\{1\}, \{2\}}\oplus\overline{W}_{\{2\}, \{1\}}, \overline{W}_{\{1\}, \{3\}}\oplus\overline{W}_{\{3\}, \{1\}} \right\}.
 \label{eq:V1 L1}
 \end{align}
 Since user 2 is able to decode its required part $\overline{W}_{\{2\}, \{1\}}$ from message $V_{1,\mathbf{d}}(L_1)$, it only needs $\overline{W}_{\{2\}, \{3\}}$, which is recovered through the message
 \begin{align}
 V_{2,\mathbf{d}}(L_1)=\left\{\overline{W}_{\{2\}, \{3\}}\oplus\overline{W}_{\{3\}, \{2\}} \right\}.\label{eq:V2 L1}
 \end{align}
 User 3 can decode its missing parts from $V_{1,\mathbf{d}}(L_1)$ and $V_{2,\mathbf{d}}(L_1)$, and therefore, $V_{3,\mathbf{d}}(L_1)=\emptyset$. We note that, while the generation of the coded messages for sublibrary $L_1$ follows similarly to generic coded caching models with a shared common link, we assign them to  users starting from the one with the worst channel gain, as the stronger users automatically decode messages destined for  weaker users.

 \item Sublibrary $L_2$: Each user requires two subfiles from $L_2$, which can be considered as two separate demands. We can group these demands into two, with only one demand per user in each group, and deliver the demands within each group separately. One possible grouping of $L_2$ could be $\mathfrak{S}_1 = (\{1,2\}, \{1,2\}, \{1,3\})$ and $\mathfrak{S}_2=(\{1,3\}, \{2,3\}, \{2,3\})$, where $\mathfrak{S}_1$ corresponds to users 1, 2 and 3 requesting subfiles $\overline{W}_{\{1,2\}}$, $\overline{W}_{\{1,2\}}$ and $\overline{W}_{\{1,3\}}$, respectively. 
 Then $V_{k,\mathbf{d}}(L_2)=\{v_{k}^1 , v_{k}^2$\}, where $v_{k}^i$ is user $k$'s message corresponding to group $\mathfrak{S}_i$, $i=1,2$. Then, for $\mathfrak{S}_1$ we have
  \begin{align}
 &v^1_1= 
 \{\overline{W}_{\{1,2\}, \{2\}}\oplus\overline{W}_{\{1,2\}, \{1\}},
 \overline{W}_{\{1,2\}, \{3\}}\oplus\overline{W}_{\{1,3\}, \{1\}}\},\label{eq:V1 L2 a}\\
 &v^1_{2}=\left\{\overline{W}_{\{1,3\}, \{2\}}\oplus\overline{W}_{\{1,2\}, \{3\}} \right\},\label{eq:V2 L2 a}\\
 & v^1_{3}=\emptyset,\label{eq:V2 L2 c}
 \end{align}
 and for $\mathfrak{S}_2$
 \begin{align}
  &v^2_1=
 \{\overline{W}_{\{1,3\}, \{2\}}\oplus\overline{W}_{\{2,3\}, \{1\}},
 \overline{W}_{\{1,3\}, \{3\}}\oplus\overline{W}_{\{2,3\}, \{1\}}\},\label{eq:V1 L2 b}\\
 &v^2_{2}=\left\{\overline{W}_{\{2,3\}, \{2\}}\oplus\overline{W}_{\{2,3\}, \{3\}} \right\},\label{eq:V2 L2 b}\\
 & v^2_{3}=\emptyset.\label{2eq: V2 L2 b}
 \end{align}
 \item Sublibrary $L_3$: All users require $\overline{W}_{\{1,2,3\}}$, and therefore
  \begin{align}
 &V_{1,\mathbf{d}}(L_3)=
  \{\overline{W}_{\{1,2,3\}, \{2\}}\oplus\overline{W}_{\{1,2,3\}, \{1\}}, \overline{W}_{\{1,2,3\}, \{3\}}\oplus\overline{W}_{\{1,2,3\}, \{1\}} \},\label{eq:V1 L3}\\
 &V_{2,\mathbf{d}}(L_2)=  V_{3,\mathbf{d}}(L_2)=\emptyset. \label{eq:V2 L3}
 \end{align}
 \end{itemize}
 The messages in \eqref{eq:V1 L1}, \eqref{eq:V1 L2 a}, \eqref{eq:V1 L2 b} and \eqref{eq:V1 L3} constitute all the messages targeted for user 1, with total rate $\rho_1=2(R_1+2R_2+R_3)$. Messages \eqref{eq:V2 L1}, \eqref{eq:V2 L2 a} and \eqref{eq:V2 L2 b} are targeted for user 2 with total rate $\rho_2=R_1+2R_2$, and finally, user 3 can successfully recover its requested file from the messages intended for users 1 and 2, i.e., $\rho_3=0$. Based on Proposition~\ref{prop:AWGN}, the target rates can be delivered to the users with superposition coding of Gaussian codewords  satisfying \eqref{eq:codebook}, with a minimum power value given in \eqref{eq:power}.

\subsection{Proposed Scheme}\label{sec:scheme general}
 This section presents the  proposed centralized caching and delivery scheme, which generalizes Example 2 to an arbitrary number of users, and achieves the transmit power value claimed in Theorem~\ref{thm:ach power}. Similarly to the schemes in \cite{qiandeniz2018icc,hassanzadeh2017rate,hassanzadeh2017broadcast}, the proposed scheme treats the sublibraries independently: 1) the cache capacity is divided among $N$ sublibraries, 2) for each demand realization, the server identifies the messages that need to be delivered to each user, independently across sublibraries, using a modified version of the scheme proposed in \cite{qiandeniz2018icc}, and 3) the server employs superposition coding to reliably communicate coded messages over the Gaussian BC.

\subsubsection{\bf Placement Phase}\label{placementphase}
Cache contents are identified separately for different sublibraries, each with a  different level of commonness. Let $\boldsymbol \pi = (\pi_1,\dots,\pi_N)$ denote the cache allocation vector, where $\pi_\ell\in[0,1]$ denotes the fraction of the normalized cache capacity $M$ allocated to sublibrary $L_\ell$, with $\sum_{\ell=1}^N \pi_\ell=1$. We will later optimize $\boldsymbol \pi$ to minimize the required total power. For a given $\boldsymbol \pi$, placement for sublibrary $L_\ell$ is carried out using the prefetching scheme proposed in~\cite{yu2017exact} as follows. Let
\begin{equation}\label{tl}
t_\ell\triangleq \frac{K\pi_\ell M}{\binom{N}{\ell}R_\ell}, ~  t_\ell\in[0,K],
\end{equation}
which is not necessarily an integer. We address this by memory-sharing among neighboring integer points, $t_\ell^A \triangleq \lfloor t_\ell \rfloor$ and $t_\ell^B \triangleq\lfloor t_\ell \rfloor+1$, and divide each subfile $\overline{W}_{\mathcal{S}}\in L_\ell$ into two non-overlapping parts. More specifically, $\overline{W}_{\mathcal{S}}=(\overline{W}^A_{\mathcal{S}}, \overline{W}^B_{\mathcal{S}})$, where  $\overline{W}^A_{\mathcal{S}}$ is at rate $(t_\ell^B-t_\ell)R_\ell$, while $\overline{W}^B_{\mathcal{S}}$ is at rate $(t_\ell- t_\ell^A)R_\ell$. The prefetching policy of \cite{yu2017exact} is implemented separately for $\{\overline{W}^A_{\mathcal{S}}: \mathcal{S}\in L_\ell\}$ and $\{\overline{W}^B_{\mathcal{S}}: \mathcal{S}\in L_\ell\}$. Each  $\overline{W}^A_{\mathcal{S}}$ is split into $\binom{K}{t_l^A}$ non-overlapping equal-length parts, each of size $n( t_\ell^B-t_\ell)R_\ell/\binom{K}{ t_\ell^A} $ bits. These parts are assigned to sets $\mathcal A \subseteq [K]$ of size $|\mathcal A| = t_\ell^A$. We denote the part assigned to set  $\mathcal A$ by $\overline{W}^A_{\mathcal{S}, \mathcal{A}}$; therefore,
\begin{equation}
\overline{W}^A_{\mathcal{S}}= \{   \overline{W}^A_{\mathcal{S}, \mathcal{A}} :\mathcal{A}\subseteq [K] ,\, |\mathcal{A}|=  t_\ell^A \}. \notag
\end{equation}
   Similarly, each  $\overline{W}^B_{\mathcal{S}}$ is split into $\binom{K}{t_l^B}$ non-overlapping equal-length parts, which are labeled as 
\begin{equation}
\overline{W}^B_{\mathcal{S}}= \{   \overline{W}^B_{\mathcal{S}, \mathcal{B}} :\mathcal{B}\subseteq [K] ,\, |\mathcal{B}|=  t_\ell^B \}. \notag
\end{equation}

User $k$ caches parts  $\overline{W}^A_{\mathcal{S}, \mathcal{A}}$   if $k \in \mathcal{A}$, and parts $\overline{W}^B_{\mathcal{S}, \mathcal{B}}$   if $k \in \mathcal{B}$. With this placement strategy,  for each subfile in sublibrary $L_\ell$,  $\binom{K-1}{ t_\ell^A-1}$ distinct parts from $\overline{W}^A_{\mathcal{S}}$, and $\binom{K-1}{t_\ell^B-1}$ distinct parts from  $\overline{W}^B_{\mathcal{S}}$, are placed in each user's cache, amounting for a total of $nt_\ell R_\ell/K$ bits, which satisfies the capacity constraint of $n\pi_\ell M$ bits. 

\subsubsection{\bf Delivery Phase}\label{deliveryphase}
Delivering a file from a library of correlated files can be considered as a multiple-demand problem \cite{qiandeniz2018icc,hassanzadeh2017rate,hassanzadeh2017broadcast}. For demand vector $\bf d$, user $k$ needs $ \binom{N-1}{\ell-1}$ subfiles from sublibrary $L_\ell$. Since the sublibraries are treated independently, message $V_{k,\mathbf{d}}$, targeted at user $k$, constitutes the messages computed from all the sublibraries, i.e.,
\begin{equation}
V_{k,\mathbf{d}} =   \bigcup\limits_{\ell=1}^N V_{k,\mathbf{d}}(L_\ell),\label{eq:messages}
\end{equation}
where  $V_{k,\mathbf{d}}(L_\ell)$ denotes the set of messages from sublibrary $L_\ell$ targeted at user $k$. They are determined using Algorithm~1, which is based on \cite[Algorithms 1, 2]{qiandeniz2018icc}. The main idea is to treat subfiles $\{\overline{W}_{\mathcal S}: d_k \in \mathcal S\}$ that are not cached at user $k$, as different demands. The algorithm operates by partitioning all the requested subfiles from sublibrary $L_\ell$ into groups, such that each user requires at most one subfile in each group; resulting in a single-demand problem. 
\begin{algorithm}[H]\label{message:1}
\caption{Generate messages $\{V_{1,\mathbf{d}}(L_\ell),\dots, V_{K,\mathbf{d}}(L_\ell)\}$}
\label{groupingscheme}
\begin{algorithmic}[1]
\Statex
\small
\State{$V_{k,\mathbf{d}}(L_\ell) \leftarrow  \emptyset$, $\forall k \in \{1,\dots,K\}$}
\For{$r =1 ,\dots,\ell$}
\State{\small ${\mathcal W}_r= \{\overline{W}_{\mathcal{S}}: |S| = \ell,\; |\mathcal{S}\cap\mathcal{D}|=r \}$}
\State{$\mathfrak{S}_1,\dots, \mathfrak{S}_g$ $\leftarrow$ Group (${\mathcal W}_r$, $\mathcal{D}$, $\ell$, $r$)}  
\For{$i \in \{1 ,\dots,g\}$}
\State{$V^A_{1},\dots, V^A_K$ $\leftarrow$ Single-Demand ($A$, ${\mathfrak{S}_i}$, $t_\ell^A$)}
\State{$V^B_1,\dots, V^B_K$ $\leftarrow$ Single-Demand ($B$, ${\mathfrak{S}_i}$, $t_\ell^B$)}
\State{$V_{k,\mathbf{d}}(L_\ell)\leftarrow V_{k,\mathbf{d}}(L_\ell) \cup \{ V^A_k,V^B_k\}$, $\forall k \in \{1 ,\dots,K\}$}
\EndFor
\EndFor
\end{algorithmic}
\end{algorithm}
\begin{algorithm}
\begin{algorithmic}[1]
\Function {Group }{ ${\mathcal W}$, $\mathcal{D}$, $\ell$, $r$}\\
\textbf{Output:} Group demands $\mathfrak{S}_1,\dots, \mathfrak{S}_g$
\small
\State{$\mathcal{F} \leftarrow \mathcal{D}$, $\overline{\mathcal{F}}\leftarrow \emptyset$, $\overline{\mathcal{S}}\leftarrow \emptyset$,
$g=0$}
\While{$\mathcal{W}\neq \emptyset$}
\While{$\mathcal{F}\neq \emptyset$}
\If{$|\mathcal{F}|\geq r$}
\If{$\overline{\mathcal{F}}=\emptyset$}
\State{ Randomly pick \footnotesize{$\overline{W}_{\mathcal{S}}\in \mathcal{W}$} such that \footnotesize{$\mathcal{S}\cap{\mathcal{D}} \subseteq\mathcal{F}$}}
\State{
$\mathcal{W}\leftarrow \mathcal{W}/\overline{W}_{\mathcal{S}},\quad
\mathcal{F} \leftarrow  \mathcal{F}\setminus \mathcal{S}$}
\For{$d_k \in \mathcal{S}\cap{\mathcal{D}}$}
\State{$\mathcal{S}_{k} \leftarrow  \mathcal{S}$}
\EndFor
\Else {\For{$d_k \in \overline{\mathcal{F}}$}
\State{$\mathcal{S}_{k} \leftarrow  \overline{\mathcal{S}}$}
\EndFor}
\State{\small
$\mathcal{F} \leftarrow \mathcal{F}\setminus \overline{\mathcal{F}},\quad \overline{\mathcal{S}} \leftarrow  \emptyset,\quad \overline{\mathcal{F}} \leftarrow \emptyset, $}
\EndIf
\Else\State{\small Randomly pick $\overline{W}_{\mathcal{S}}\in \mathcal{W}$ such that $\mathcal{F}\subseteq \mathcal{S}$}
\For{$d_k \in \mathcal{F}$}
\State{$\mathcal{S}_{k} \leftarrow  \mathcal{S}$}
\EndFor
\small
\State{$\mathcal{F} \leftarrow \emptyset,\quad
 \overline{\mathcal{S}} \leftarrow  \mathcal{S},\quad \overline{\mathcal{F}} \leftarrow  \mathcal{S}\setminus  {\mathcal{F}}$}
\normalsize
\EndIf
\EndWhile
\State{$g=g+1$}
\State{$\mathfrak{S}_g= (\mathcal{S}_{1}, \dots, \mathcal{S}_{K})$}
\EndWhile
\EndFunction
\end{algorithmic}
\end{algorithm}

\begin{algorithm}
\begin{algorithmic}[1]
\Function{Single-Demand }{$C$, ${\mathfrak{S}}$, $t$}\\
\textbf{Input:} $\mathfrak{S}=(\mathcal{S}_{1}, \dots, \mathcal{S}_{K})$, $C\equiv \{\overline{W}^C_{\mathcal S,\mathcal C}\}$\\
\textbf{Output:} Coded messages $V_1,\dots,V_K$
\small
\State{$\mathcal{K} \leftarrow \{k: \mathcal{S}_{k}\notin \{\mathcal{S}_{1}, ..., \mathcal{S}_{k-1}\}\}$}
\For{$k\in \{1,\dots,K\}$}
\For{{\footnotesize$\mathcal{U}\subseteq [k+1:K]: |\mathcal{U}|= t, \sum\limits_{j \in \mathcal{K}}\mathbbm{1}\{j \in \mathcal{U}\cup\{k\}\} \geq 1$}}
\State{  $
V_k\leftarrow V_k \bigcup\left(\bigoplus\limits_{j\in \mathcal{U}\cup\{k\}} \overline{W}^C_{\mathcal{S}_{j}, \mathcal{U}\cup\{k\}\setminus \{j\}}\right)$}
\EndFor
\EndFor
\EndFunction

\end{algorithmic}
\end{algorithm}

For sublibrary $L_\ell$ messages, $V_{1,\mathbf{d}}(L_\ell),\dots, V_{K,\mathbf{d}}(L_\ell)$, are generated as follows:
  \begin{itemize}
      \item[$i)$] Group the requested subfiles: Let $\mathcal{D}\triangleq\{d_1, ..., d_K\}$ denote the set of distinct demands in $\bf d$. The subfiles that need to be delivered to at least $\ell$ users, are given by: 
      \begin{equation}\label{eq:set for empty}
      \{\overline{W}_{\mathcal{S}}: \mathcal{S} \subseteq\mathcal{D},\, |S| = \ell\}.
      \end{equation}
       Since each user can request multiple subfiles from \eqref{eq:set for empty}, they are grouped into multiple (possibly overlapping) sets with minimum cardinality, such that each group represents the demand set of a single-demand network with $K$ users, i.e., each user has a single demand within this group. The grouping process tries to minimize the number of distinct demands within each  single-demand network. For sublibrary $L_\ell$, where each subfile is required by $\ell$ distinct users, there are at most $\lceil |\mathcal{D}| /\ell\rceil+1$ subfiles in each group. Note that, the subfiles in \eqref{eq:set for empty} are not the only contents that need to be delivered from sublibrary $L_\ell$. Based on the demand vector, any subfile $\overline{W}_{\mathcal{S}}$ whose index $\mathcal{S}$ includes at least one of the indices in $\mathcal D$, i.e.,  $\mathcal{S}\cap\mathcal{D}\neq \emptyset$, is required for the lossless reconstruction of the corresponding requested file in $\mathcal{D}$. All such subfiles need to be identified, and grouped in a similar fashion. Subfiles in \eqref{eq:set for empty} correspond to $|\mathcal{S}\cap\mathcal{D}|=\ell$. For $r=1,\dots,\ell$, we define the requested subfiles  ${\mathcal W}_r$, as 
      \begin{equation}\label{eq:groups}
      {\mathcal W}_r \triangleq \{\overline{W}_{\mathcal{S}}: |\mathcal{S}| = \ell,\; |\mathcal{S}\cap\mathcal{D}|=r \}.\notag
      \end{equation}
      Then, each set ${\mathcal W}_r$ is grouped using the function GROUP in Algorithm~1, which assigns a demand vector $\mathfrak{S}_i=({\mathcal S}_1,\dots,{\mathcal S}_K)$ to each group, resulting in a single-demand network with $K$ users, where user $k$ requests subfile $\overline{W}_{\mathcal{S}_k}$.

      \item[$ii)$] Deliver the demands corresponding to each group: The groups formed above are treated independently in the delivery phase. More specifically, for a group with corresponding demand vector $\mathfrak{S}$, function SINGLE-DEMAND in Algorithm~1 identifies messages $V_1,\dots,V_K$ that need to be transmitted so that all the users recover their requested subfiles in $\mathfrak{S}$. These messages are computed using the scheme in \cite{yu2017exact}, and delivered over the degraded BC using the coding scheme in \cite{MohammadDenizerasureTCom}. The channel is taken into account by selecting the {\em weakest} users with distinct demands as {\em leaders}, i.e., the demand of a leader is not requested by any of the weaker users, $ \{k: \mathcal{S}_{k}\notin \{\mathcal{S}_{1}, ..., \mathcal{S}_{k-1}\}\}$, and then greedily broadcasting XORed messages that benefit at least one leader through superposition coding. Note that choosing the weakest user, among users requiring the same subfile $\overline{W}_{\mathcal S}$, as the leader, allows all the stronger users to decode the subfile through successive cancellation decoding. As mentioned previously, the proposed scheme uses memory-sharing to cache and deliver the subfiles in $L_\ell$, for the two parts $\overline{W}_{\mathcal{S}}^A$ and $\overline{W}_{\mathcal{S}}^B$; and therefore, function SINGLE-DEMAND is executed separately for both parts. 
   \end{itemize}
   
  Message $V_{k,\mathbf{d}}(L_\ell)$ targeted at user $k$ is the union of all the messages for sublibrary $L_\ell$ computed for each group identified from the subfile sets $\{\mathcal W_1,\dots, \mathcal W_\ell\}$, from which the overall message for user $k$, $V_{k,\mathbf{d}}$, is obtained by \eqref{eq:messages}. For a given demand vector $\bf d$, messages $V_{1,\mathbf d},\dots,V_{K,\mathbf d}$ can be reliably transmitted to users $1,\dots,K$, using a $K$-level Gaussian superposition codebook \cite{Bergmans1974,bergmans1973degradedBC}. The $k^\text{th}$-level codebook consists of $2^{n\rho_{k}}$ codewords, where $\rho_k$ is the total rate of the messages in $V_{k,\mathbf{d}}$. The total required transmit power is given by \eqref{eq:codebook} in Proposition \ref{prop:AWGN}.

\subsection{Achievable transmit power}
The worst-case transmit power of the scheme described above is presented next. 

  \begin{Theorem}\label{thm:ach power}
For the caching problem described in Section~\ref{sys}, the optimal memory-power function, $P^*(M)$, is upper bounded as
\small
\begin{align}
P^*(M)\leq  &\min\limits_{ {\boldsymbol \pi} = (\pi_1,\dots,\pi_N)} P_{UB}(M,  {\boldsymbol\pi}),\notag\\
&\qquad \mathrm{s.t.~}\quad   \sum\limits_{i=1}^N \pi_i\leq 1,\notag\\
&\qquad\qquad\quad  0\leq \pi_i \leq 1, ~~ i = 1,\dots,N,  \notag
\end{align}
\normalsize
where
\small
\begin{subequations}
\begin{align}
& P_{UB}(M,  {\boldsymbol\pi})\triangleq\sum\limits_{k=1}^{ K}\left(\frac{2^{2\hat{\rho}_k}-1}{h_k^2}\right) \prod\limits_{j=1}^{k-1}{2^{2\hat{\rho}_j}},\notag \\
&\hat{\rho}_k\triangleq \sum\limits_{\ell=1}^N \sum\limits_{r=\max\{\ell-N+K, 1\}}^{\min\{\ell, K\}  }\binom{N-K}{\ell-r}\binom{\min\{N, K\}-1}{r-1}\gamma_{k, \ell, r},\notag\\
& {\gamma}_{k,\ell, r}\triangleq \begin{cases}
  \Big( \frac{\binom{K-k}{\lfloor t_\ell\rfloor}}{\binom{K}{\lfloor t_\ell \rfloor}}(\lfloor t_\ell \rfloor+1-t) +\frac{\binom{K-k}{\lfloor t_\ell\rfloor+1}}{\binom{K}{\lfloor t_\ell\rfloor+1}}(t-\lfloor t_\ell \rfloor)\Big) R_\ell,\, \mathrm{if}\; k\in[   \lceil\frac{\min\{N, K\}}{r}\rceil +1 ], \\
\qquad\qquad0\qquad\qquad\qquad\quad\qquad\qquad\qquad\quad \mathrm{otherwise}
\end{cases}\notag\\
& t_{\ell} \triangleq \frac{K \pi_\ell M}{\binom{N}{\ell}R_\ell}.\notag
\end{align}
\end{subequations}
\normalsize
\begin{proof}
This transmit power is achieved by the coding scheme outlined in Algorithm~1. A detailed proof is given in Appendix B.
\end{proof}
\end{Theorem}

\section{Coded Placement and Joint Encoding}\label{sec:scheme joint}
We propose an alternative joint cache-channel coding scheme, with coded placement which is more effective for small cache sizes. The scheme operates by constructing a multi-level superposition code, based on the demand realization, and piggyback part of the messages targeted at each user on the messages intended for weaker users. The piggyback coding is also employed in \cite{shirin2017broadcastit}, where all the cache capacity allowance is assigned to the weakest user, and in the delivery phase, part of the content required by each user is piggy-backed onto the message sent to the weakest user. We extend this scheme in two ways: the coded placement is implemented instead of uncoded placement, and the piggyback coding is applied to each layer of superposition code instead of just the first layer. Before presenting the general scheme description, we first provide a brief overview of cache-aided superposition coding, and then use an example to illustrate how part of the messages required by a stronger user can be piggy-backed onto the messages targeted at weaker users.

\subsection{Preliminaries}\label{sec:superposition}
 We extend the piggyback coding in \cite{shirin2017broadcastit} to the case when each user has cached contents. In a cache-aided $K$-user degraded Gaussian BC with $h_1^2\leq h_2^2\leq\dots\leq h_K^2$, where message $V_k^r$, with rate $\rho^r_k$, is locally available at user $k\in[K]$, message $V_k^c$, with rate $\rho_k^c$, can be reliably transmitted to user $k$, and message $V_k = (V_k^r, V_k^c) $, with rate $\rho_k =\rho_k^r + \rho_k^c $, can be decoded by users $k+1,\dots,K$, using $K$-level superposition coding as follows:
 
 \begin{itemize}
 \item \textit{Codebook construction}: The $k^\text{th}$ level codebook, denoted by $\mathcal C_k$, consists of $\lfloor 2^{n \rho_k^r}\rfloor \times \lfloor 2^{n \rho_k^c}\rfloor$ codewords of block length $n$, denoted by $x_k^n(v_k^r, v_k^c)$, $v_k^r \in \lfloor 2^{n \rho_k^r}\rfloor$, $v_k^c \in \lfloor 2^{n \rho_k^c}\rfloor$, which are arranged into $\lfloor 2^{n \rho_k^r}\rfloor$ rows and $\lfloor 2^{n \rho_k^c}\rfloor$ columns. The codewords in $\mathcal C_k$ are generated independently and identically distributed (i.i.d.) following $x_{k,i}\sim \mathcal{N}(0,P_k)$, $i\in [n]$.
 
\item \textit{Encoding at server}: For messages $V_1,\dots,V_K$ targeted at the users, the server transmits the superposition of the $K$ codewords $\sum\limits_{k=1}^K x_k^n(V_k^r, V_k^c)$ over the Gaussian BC.

\item \textit{Decoding at users}: User $k\in[K]$ receives the channel output 
\begin{equation}
Y^n_{k}=h_k\sum\limits_{k=1}^K x_k^n(V_k^r, V_k^c)+\sigma^n_{k},\notag
\end{equation}
and based on Proposition~\ref{prop:AWGN}, it can successfully decode messages $V_1, \dots, V_{k-1}$, by using successive decoding if 
\small
\begin{equation}\label{rccondition}
 \rho_j^r + \rho_j^c \leq C\left(\frac{h_k^2P_j}{1+h_k^2\sum\limits_{j'=k+1}^{K}P_{j'}}\right), ~~~~\forall j\in [k-1].
\end{equation}
\normalsize
Since user $k$ has access to $V^k_{r}$, it can extract the subcodebook $\Big\{x_k^n(V^r_{k},v^c_{k}): v^c_{k} \in \lfloor2^{n \rho_k^c}  \rfloor \Big\}$ from $\mathcal C_k$, and losslessly decode $V^c_{k}$, if
\small
\begin{equation}
 \rho_k^c \leq C\left(\frac{h_k^2P_k}{1+h_k^2\sum\limits_{j=k+1}^{K}P_j}\right).\notag
\end{equation}
\end{itemize}
\normalsize
If \eqref{rccondition} holds, users $k,\dots K$ can also decode messages $V_1, ..., V_{k-1}$ as they have better channel conditions. However, they do not have access to side information $V_k^c$, so for them to decode $V_k$ successfully, we need 
\small
 \begin{align}
   \rho_k^r + \rho_k^c  &\leq C\left(\frac{h_{k+1}^2P_k}{1+h_{k+1}^2\sum\limits_{j=k+1}^{K}P_{j}}\right).\notag
\end{align}
\normalsize

\textbf{Example 3}. Consider the model in Example 2, with the rates of the subfiles in three commonness levels given by $R_3\leq R_2\leq R_1$, and normalized cache capacity of $M = R_3$.

\noindent $\circ$ \textbf{Placement Phase:} First, we divide each of the subfiles in $L_1$ and $L_2$ into two parts: 
\begin{itemize}
    \item Sublibrary $L_1$: $\overline{W}_{\{i\}}=(\overline{W}^C_{\{i\}}, \overline{W}^U_{\{i\}})$, $i\in [3]$, where $\overline{W}^C_{\{i\}}$ has rate $R_3$ while $\overline{W}^U_{\{i\}}$ has rate $R_1-R_3$.
    \item Sublibrary $L_2$:  $\overline{W}_{\mathcal{S}}=(\overline{W}^C_{\mathcal{S}}, \overline{W}^U_{\mathcal{S}})$, $\overline{W}_{\mathcal{S}}\in L_2$, where $\overline{W}^C_{\mathcal{S}}$ has rate $R_3$ while $\overline{W}^U_{\mathcal{S}}$ has rate $R_1-R_3$.
\end{itemize}
Then, users 1, 2, and 3 cache coded contents as follows:
\small
\begin{align}
Z_1 &=  \overline{W}_{\{123\}}, \notag\\
Z_2 & =  \overline{W}^C_{\{12\}}\oplus \overline{W}^C_{\{13\}}\oplus \overline{W}^C_{\{23\}}, \notag\\
Z_3 & =  \overline{W}^C_{\{1\}}\oplus \overline{W}^C_{\{2\}}\oplus \overline{W}^C_{\{3\}},\notag
\end{align}
\normalsize
such that the weaker users prefetch a linear combination of the subfiles shared among more files.

\noindent $\circ$ \textbf{Delivery Phase:}
\begin{itemize}
    \item Codebook construction: For the demand vector $\mathbf{d}=(1, 2, 3)$, as explained in Section ~\ref{sec:superposition}, to apply piggyback coding, the server generates a 3-level Gaussian superposition codebook as follows:
   \begin{itemize}
     \item[-] $\mathcal C_1$ with $\lfloor 2^{nR_3}\rfloor$ rows and $\lfloor2^{n(R_1+2R_2)}\rfloor$ columns,
     \item[-] $\mathcal C_2$ with $\lfloor 2^{nR_3}\rfloor$ rows and $\lfloor2^{n(R_1+R_2-R_3)}\rfloor$ columns,
     \item[-] $\mathcal C_3$ with $\lfloor 2^{nR_3}\rfloor$ rows and $\lfloor2^{n(R_1-R_3)}\rfloor$ columns,
 \end{itemize}
which contain i.i.d. codewords of length $n$ generated from zero-mean Gaussian distributions with variances $P_1$, $P_2$, and $P_3$, respectively.
    
    \item Encoding at server: The server transmits 
    \small
    \begin{align}
    X^n({\mathbf W},{\mathbf d}) = x_1^n(V_{1,\bf d}^r,V_{1,\bf d}^c) + x_2^n(V_{2,\bf d}^r,V_{2,\bf d}^c) + x_3^n(V_{3,\bf d}^r,V_{3,\bf d}^c), \notag
    \end{align}
    \normalsize
    where
    \small
    \begin{align}
    \hspace{0.5cm} & V_{1,\bf d}^r =Z_1, 
    & &V_{1,\bf d}^c  = (  \overline{W}_{\{1\}},
    \overline{W}_{\{12\}},
    \overline{W}_{\{13\}}   ), \hspace{2cm} \notag\\ 
    &V_{2,\bf d}^r =Z_2, 
    & &V_{2,\bf d}^c  = (  \overline{W}^U_{\{23\}}, \overline{W}_{\{2\}}  ), \notag\\
   & V_{3,\bf d}^r = 1,  
   & &V_{3,\bf d}^c  = \overline{W}^U_{\{3\}} .\notag
 \end{align}
    \normalsize
    \item Decoding at users: 
     \begin{itemize}
    \item User 1 has the weakest channel gain and needs to receive all the subfiles it has not prefetched, i.e., $\Big\{\overline{W}_{\{1\}},
 \overline{W}_{\{12\}},$ $
 \overline{W}_{\{13\}}\Big\}$. Using its cached content $\overline{W}_{\{123\}}$, it can extract the subcodebook $\Big\{x_1^n(\overline{W}_{\{123\}} ,\, v^c_{1}):$ $ v^c_{1} \in [   2^{n(R_1+2R_2)}  ] \Big\}$ from $\mathcal C_1$ and losslessly recovers the required parts if 
 \small
\begin{equation}
 R_1+2R_2 \leq C\left(\frac{h_1^2P_1}{1+h_1^2(P_2+P_3)}\right).\label{eq:piggy first}
\end{equation}
\normalsize
\item User 2 requires $\{\overline{W}_{\{2\}},
 \overline{W}_{\{12\}},
 \overline{W}_{\{23\}}, \overline{W}_{\{123\}} \}$, and if
 \small
 \begin{equation}\label{piggyback2}
 R_1+2R_2+R_3 \leq C\left(\frac{h_2^2P_1}{1+h_2^2(P_2+P_3)}\right),
\end{equation}
\normalsize
it can first decode $\overline{W}_{\{123\}}$, $\overline{W}_{\{12\}}$ and $\overline{W}_{\{13\}}$ from the codebook $x_1^n$, and can retrieve $\overline{W}^C_{\{23\}}$ from its cached contents. It can then decode the remaining parts required to reconstruct file $W_2$, i.e., parts $\overline{W}_{\{2\}}$ and $\overline{W}^U_{\{23\}}$ from $x_2^n$ using its side information $Z_2$ if
\small
 \begin{equation}
 R_1+R_2-R_3 \leq C\left(\frac{h_2^2P_2}{1+h_2^2\, P_3}\right).
\end{equation}
\normalsize
    
\item User 3 can decode messages $\{\overline{W}_{\{1\}},\overline{W}_{\{12\}}, \overline{W}_{\{13\}}, \overline{W}_{\{123\}}\}$ from $x_1^n$ if \eqref{piggyback2} is satisfied, since $h_3\geq h_2$,
and decode messages $\{\overline{W}_{\{2\}},\overline{W}^U_{\{23\}}, \overline{W}^C_{\{12\}}\oplus \overline{W}^C_{\{13\}}\oplus \overline{W}^C_{\{23\}}\}$ from $x_2^n$, if
\small
\begin{equation}
R_1+R_2 \leq C\left(\frac{h_2^2P_2}{1+h_2^2\, P_3}\right).
\end{equation}
\normalsize
With $\overline{W}_{\{12\}}, \overline{W}_{\{13\}}$, it can decode $\overline{W}^C_{23}$ using the coded side information in its cache. Then, only subfile $\overline{W}_{\{3\}}$ is left for user 3 to fully recover $W_3$. To this end, it can recover $\overline{W}^C_{\{3\}}$ from its cache as it has already decoded $\overline{W}_1^c$ and $\overline{W}_2^c$. Finally, it can decode $\overline{W}^U_{\{3\}}$ from $x_3^n$ if
\small
 \begin{equation}
 R_1-R_3 \leq C\left( h_3^2P_3\right). \label{eq:piggy last}
\end{equation}
\normalsize
\end{itemize}
\end{itemize}
    
 The transmission powers $P_1, P_2, P_3$ are chosen to satisfy Eqs~\eqref{eq:piggy first}-\eqref{eq:piggy last}. As it can been seen from the example, the idea is to jointly encode the cached contents of each user together with the message intended for it. This additional message does not interfere with the weak user as it already has it cached, while the stronger users can recover this information without any additional transmission cost.

\subsection{Proposed Scheme}\label{subsec:piggy general}
We now present the proposed coded caching and joint encoding scheme for a general setting with $N\geq K$, and a normalized cache capacity $M\leq$ $ \min\{$ $R_{N-K+1},$ $\dots,$ $R_N\}$. We will explain later how the scheme can be applied to arbitrary number of users and files.

\subsubsection{\bf Placement Phase}
Each subfile $\overline{W}_{\mathcal{S}}$, $\mathcal{S}\subseteq [N]$,  is divided into two non-overlapping parts,  $\overline{W}_{\mathcal{S}}=(\overline{W}^{C}_{\mathcal{S}}, \overline{W}^{U}_{\mathcal{S}})$, where $\overline{W}^{C}_{\mathcal{S}}$ is at rate $M$, and  $\overline{W}^{U}_{\mathcal{S}}$ is at rate  $R_{|\mathcal{S}|}-M$.  User $k\in[K]$ caches a linear  combination of all the parts $\overline{W}^C_{\mathcal{S}}$ in sublibrary $L_{N-k+1}$ as 
\small
\begin{equation}\label{cachedop}
Z_k=\bigoplus \limits_{ \mathcal S \subseteq [N]   :\,      |\mathcal S| = N-k+1}\overline{W}^C_{\mathcal{S}}, 
\end{equation}
\normalsize
which satisfies the cache capacity constraint $M$. 

\subsubsection{\bf Delivery Phase} 
For any demand vector $\mathbf{d}=(d_1, ..., d_K)\in [N]^K$, let $N_e(\mathbf{d})$ denote the number of distinct requests in demand ${\bf d}$, and let 
$\mathcal U \triangleq \{k_1, ..., k_{N_e(\mathbf{d})}\}$ denote the set of users with the weakest channels that request distinct files such that $|\mathcal U|=N_e(\mathbf{d})$, where $k_1 < \cdots < k_{N_e(\mathbf{d})}$. 

\begin{itemize} 
\item Codebook construction: The server constructs a $N_e(\mathbf{d})$-level Gaussian superposition codebook, such that for $i\in [N_e]$, the $i^\text{th}$-level codebook contains $2^{n \rho_i}$ codewords, where $ \rho_i = \sum\limits_{\ell=1}^{N-i+1}\binom{N-i+1}{\ell-1}R_{\ell}$.
If $k_i=i$, the codewords are arranged in an array of $2^{nM}$ rows and $2^{n(\rho_i-M)}$ columns; otherwise, i.e., $k_i\neq i$, they are arranged into $1$ row and $2^{n\rho_i}$ columns. For each element of the array we generate an i.i.d. codeword $x_i^n(v_i^r, v_i^c)$, $v_i^r \in[2^{nM}]$ and $v_i^c \in[2^{n(\rho_i-M)}]$ if $k_i=i$; $v_i^r =1$ and $v_i^c \in[2^{n\rho_i}]$ if $k_i\neq i$, with distribution $\mathcal{N}(0,P_i)$.

\item Encoding  at  server: The server transmits codeword $\sum\limits_{i=1}^{N_e(\bf d)} x_i^n(V_{i,\bf d}^r, V_{i,\bf d}^c)$, where, for $i\in[N_e(\bf d)]$, message
\small
\begin{align}\label{rowmessage}
&V_{i,\mathbf{d}}^r=\begin{cases}Z_{k_i}, &\text{if}\quad k_i=i,\\
\emptyset,  &\text{if}\quad k_i\neq i,
\end{cases}
\end{align}
\normalsize
is targeted at users $k_i+1,\dots,K$, and message
\small
\begin{align}\label{columnmessage}
&V_{i,\mathbf{d}}^c=\begin{cases}  \overline{W}^U_{\widetilde{\mathcal S}}\; \bigcup\; \Big\{   \overline{W}_{\mathcal{S}} \in  L_{N-i+1}:   \mathcal{S}\neq\widetilde{\mathcal S} \Big\}\;\bigcup\; \Big\{   \overline{W}_{\mathcal{S}} \notin L_{N-i+1} :   \mathcal{S}\in \mathcal D_{i} \Big\}~&\text{if}\quad k_i=i,\\
\Big\{   \overline{W}_{\mathcal{S}}  :   \mathcal{S}   \in \mathcal D_{i}  \Big\}\;~&\text{if}\quad k_i\neq i,
\end{cases}
\end{align}
\normalsize
$\text{for any }\widetilde{\mathcal S}$ such that $\overline{W}_{\widetilde{\mathcal S}} \in L_{N-i+1}$, is targeted at users $k_i,\dots,K$, where
\small
\begin{align}
&\mathcal D_i \triangleq  \Big\{\mathcal S: \mathcal S \subseteq [N]\setminus\{d_{k_1},\dots,d_{k_{i-1}}\} ,\, d_i \in \mathcal{S} ,\,      |\mathcal S| \leq N-i+1\Big\}\notag
\end{align}
\normalsize
is the set of subfiles required to reconstruct file $W_{d_{k_i}}$ requested by user $k_i$, but not common to any of the files requested by the weaker users, i.e., $W_{d_1},\dots, W_{d_{k_{i-1}}}$. Codeword $x_i^n(V_{i,\bf d}^r, V_{i,\bf d}^c)$ is generated with average power $P_i$ such that 
\small
\begin{align}
|V_{i,\mathbf{d}}^c|+|V_{i,\mathbf{d}}^r| &\leq C\left(\frac{h_{k_i+1}^2P_i}{1+h_{k_i+1}^2\sum\limits_{j=i+1}^{N_e(\bf d)}P_{j}}\right),\label{eq:cond all}\\
|V_{i,\mathbf{d}}^c| &\leq C\left(\frac{h_{k_i}^2P_i}{1+h_{k_i}^2\sum\limits_{j=i+1}^{N_e(\bf d)}P_j}\right), \label{eq:cond row}
\end{align}
\normalsize
where $|V_{i,\mathbf{d}}^c|$ and $|V_{i,\mathbf{d}}^r|$ denote the rates of $V_{i,\mathbf{d}}^c$ and $V_{i,\mathbf{d}}^r$, respectively.

\item Decoding at  users: 
\begin{itemize} 
\item For $i\in[N_e(\mathbf{d})]$, user $k_i$ decodes all its desired messages in two steps.

 {\em Step 1}: In the first step, user $k_i$ recovers all the messages $\{ V_{i',\bf d}^r,\, V_{i',\bf d}^c : i'\in [i-1]\}$, which correspond to all the subfiles required to reconstruct files $W_{d_{k_1}},\dots,W_{d_{k_{i-1}}}$, by decoding the first $i-1$ level codewords. This can be done with arbitrarily low error probability since condition \eqref{eq:cond all} is satisfied. 
 
 {\em Step 2}: We note that $V_{i,\bf d}^r$ is either in user $k_i$'s local cache or is an empty message. Thus, user $k_i$ always has the knowledge of $V_{i,\bf d}^r$, which together with \eqref{eq:cond row} is satisfied, it can allows the user to successfully decode $V_{i,\bf d}^c$. 
  
  Overall, user $k_i$ recovers the subfiles $\{\overline{W}_{\mathcal{S}}: \mathcal{S}\subseteq [N], \mathcal{S}\cap \{d_1, ..., d_{k-1}\}\neq \emptyset, d_k\in \mathcal{S}\}$ in the first step, and the subfiles  $\{\overline{W}_{\mathcal{S}}: \mathcal{S}\subseteq [N]\setminus \{d_1, ..., d_{k-1}\}, d_k\in \mathcal{S}\}$ in the second step, from which it can fully reconstruct $W_{d_k}$.

\item If $k \notin \mathcal U$, then user $k$ has requested the same file as a weaker user $k_i \in \mathcal U$, i.e., $k_i\leq k$. Therefore, user $k$ can decode all the messages targeted at user $k_i$, and since \eqref{eq:cond row} is satisfied, user $k$ can also recover $V^r_{i,\bf d}$, from which it can fully reconstruct $W_{d_k}$ .
\end{itemize}

\end{itemize}

\begin{Remark}
We consider more files than users, i.e., $N\geq K$, but the analysis for case $N<K$ follows directly. Note that, since each user stores a coded combination of all the subfiles in a sublibrary, with more users than files, i.e., $N<K$, the $K-N$ strongest users would be able to decode all of their required subfiles from the messages targeted at users $1,\dots,N$, rendering the cached contents $Z_{N+1},\dots, Z_{K}$ unutilized.
\end{Remark}

For any demand vector $\mathbf d$, the  total  transmit  power required by the proposed caching scheme can be upper bounded as in the following theorem.  
\begin{Theorem}\label{thm:piggy}
For the caching problem described in Section~\ref{sys}, with cache capacity 
\small
\[M\leq \min \Big\{  R_{\zeta}  ,\dots,\,R_{N} \Big\},\quad \zeta   \triangleq \max\{N-K,1\}\]
\normalsize
an upper bound on the optimal memory-power function, $P^*(M)$, is given by 
\small
\begin{equation}
P^*(M)\leq P^{\text{PB}}_{UB}(M)\triangleq  \sum\limits_{k=1}^K P_k(M),\notag \end{equation}
\normalsize
where
\small
\begin{align}
&P_k(M)= \notag\\
&\begin{cases}0, &\mbox{if}~ k\notin [\min\{N, K\}]\\
\max\bigg\{\Big(\frac{2^{2\tilde{\rho}_k} -1}{h^2_k}\Big)\Big(1+ h^2_k \, \sum\limits_{j=k+1}^{K}P_j \Big) ,\Big(\frac{2^{2(\tilde{\rho}_k+M)} -1}{h^2_{k+1}}\Big)\Big(1+ h^2_{k+1} \,
\sum\limits_{j=k+1}^{K}P_j \Big) \, \bigg\}, &\mbox{if}~ k\in [\min\{N, K\}],
\end{cases}\notag
\end{align}
\normalsize
with $\tilde{\rho}_k$ defined as in \eqref{rstar}.
\end{Theorem}
 \begin{proof}
The proof is given in Appendix~\ref{app:piggy}, which is derived by characterizing the transmit power achieved by the caching and delivery scheme described in Section~\ref{subsec:piggy general}.
\end{proof}
\begin{Remark}
We observe that, if
\small
\begin{equation}\label{optimalcondition}
\Big(\frac{2^{2\tilde{\rho}_k} -1}{h^2_k}\Big)\Big(1+ h^2_k \, \sum\limits_{j=k+1}^{K}P_j \Big) \geq \Big(\frac{2^{2(\tilde{\rho}_k+M)} -1}{h^2_{k+1}}\Big)\Big(1+ h^2_{k+1} \,
\sum\limits_{j=k+1}^{K}P_j \Big), \qquad\forall k \in [\min\{N, K\}],
\end{equation}
\normalsize
then $P_{UB}^{\text{PB}}(M)=P_{LB}(M)$, i.e., the transmission power required by the coded placement and joint encoding scheme meets the lower bound. However, it does not necessarily mean that the proposed scheme is optimal as the lower bound is derived assuming uncoded placement phase, while the proposed scheme caches contents in a coded manner. Nevertheless, we can conclude that the performance of the proposed scheme is no worse than the optimal scheme with uncoded placement phase.  
\end{Remark}

\section{Numerical results}\label{sec:numerical}
We evaluate the performance of the scheme proposed in Sec~\ref{sec:scheme general}, referred to as the {\em correlation-aware} scheme, by comparing its memory-power trade-off with the lower bound presented in Theorem~\ref{lowerbound}, as well as with the trade-off achieved by the scheme proposed in \cite{amiri2017gaussian}, which does not exploit the correlation among files, referred to as the {\em correlation-ignorant} scheme. In the latter scheme, we treat each file as a distinct sequence of bits. We consider a setting with $N=5$ files, $K=5$ users, file rate $R=1$, and cache capacity $M=0.5$. Channel gains are modeled as $1/h_k^2=2-0.2(k-1)$, for $k=1, ..., 5$. We denote by $\alpha_\ell$ the file-length fraction that belongs to sublibrary $L_\ell$, i.e.,
\small
\begin{equation}
\alpha_\ell=\binom{N-1}{\ell-1}\frac{R_{\ell}}{R}, \quad \sum\limits_{\ell=1}^N \alpha_\ell=1.\notag 
\end{equation}
\normalsize

\begin{figure}
\centering
\label{fig:1}
\includegraphics[width=0.65\linewidth]{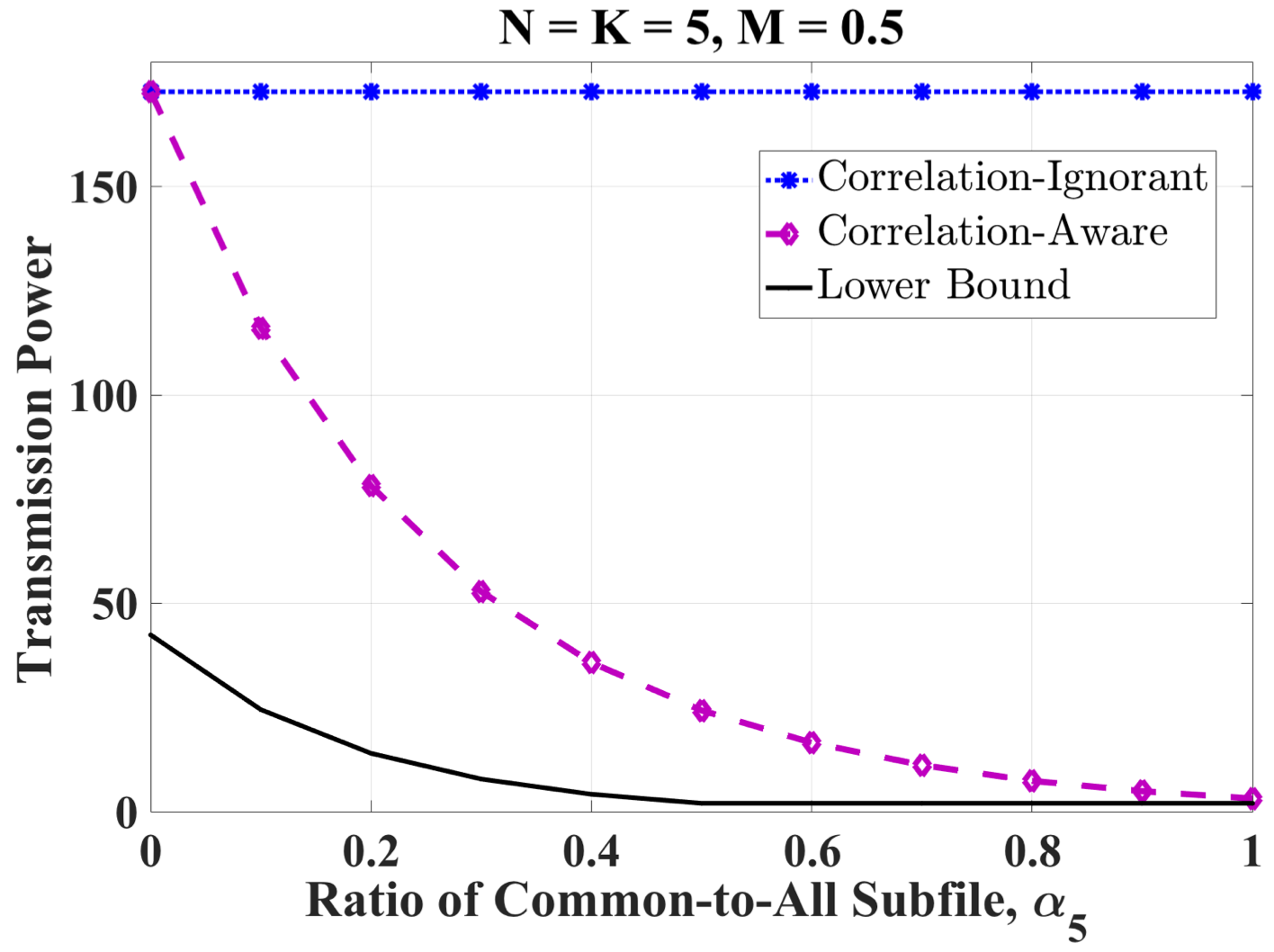}
\caption{Transmission power vs. common subfile fraction, when  the files are composed of private and common-to-all subfiles. The channel gains are given as $1/h_k^2=2-0.2(k-1)$, $k=1, ..., 5$. The correlation-aware scheme corresponds to the superposition coding scheme in Section \ref{sec:scheme general}.}\label{fig:commontoall}
\end{figure}
\begin{figure}

\label{fig:2}
\centering
\includegraphics[width=0.65\linewidth]{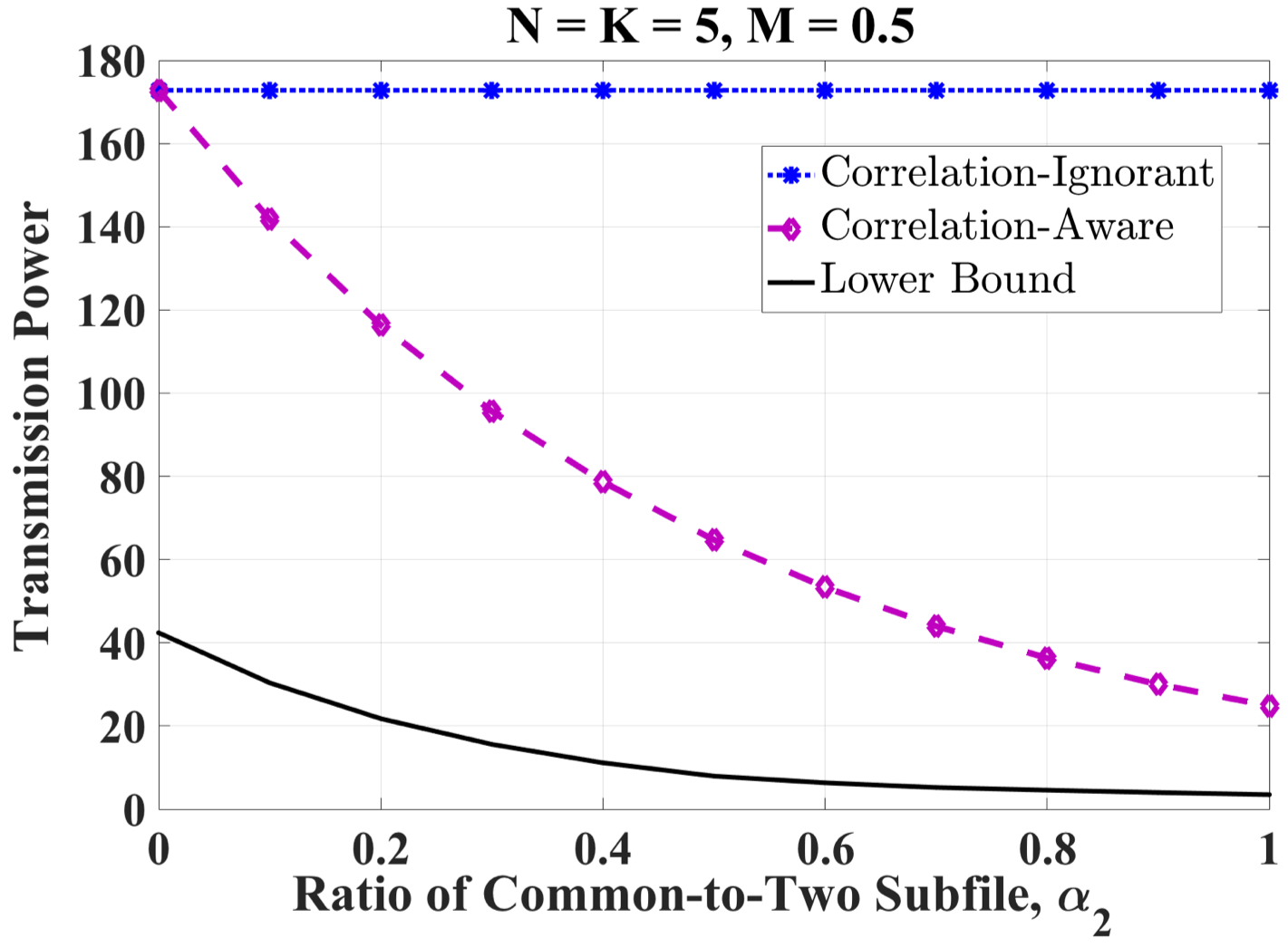}
\caption{Transmission power vs. common subfile fraction, when  the files are composed of private and common-to-two subfiles. The channel gains are given as $1/h_k^2=2-0.2(k-1)$, $k=1, ..., 5$. The correlation-aware scheme corresponds to the superposition coding scheme in Section \ref{sec:scheme general}.}\label{fig:pairwise}
\end{figure}

Fig.~\ref{fig:commontoall} displays the memory-power trade-off for a database with files composed of one {\em private} subfile, which is exclusive to that file, and a {\em common-to-all} subfile, which is shared among all the files, i.e., $\alpha_1+\alpha_5=1$, $\alpha_2=\alpha_3=\alpha_4=0$. In Fig.~\ref{fig:pairwise} the trade-off is shown when the files, in addition to private subfiles, have pairwise correlations through {\em common-to-two} subfiles, that is $\alpha_1+\alpha_2=1$, $\alpha_3=\alpha_4=\alpha_5=0$. We plot the minimum transmit power as a function of the common parts of the files for both scenarios, i.e., with respect to $\alpha_5$ and $\alpha_2$, respectively. In both settings the transmission power achieved by the correlation-aware scheme decreases remarkably, as the portion of common subfiles  increases, while the performance of the  correlation-ignorant scheme does not improve. It is observed that the transmission power drops faster in Fig.~\ref{fig:commontoall} compared to Fig.~\ref{fig:pairwise} for increasing ratio of common subfiles, in both the correlation-aware scheme and the lower bound. This is due to the reduction in the amount of content that needs to be sent over the Gaussian BC for a higher level of correlation among the files. For example, in Fig. \ref{fig:commontoall}, as $\alpha_5$ approaches $1$, all the files become the same, and hence, only a message of rate $R/2$ needs to be multicasted to all the users, whereas in the setting of Fig. \ref{fig:pairwise}, with $\alpha_2=1$, we still have $\binom{N}{2}=10$ distinct subfiles each shared by only two files. It is also observed that the gap between the transmit power upper and lower bounds is smaller in Fig.~\ref{fig:commontoall} compared to Fig.~\ref{fig:pairwise}.

\begin{figure}
\centering
\includegraphics[width=0.75\linewidth]{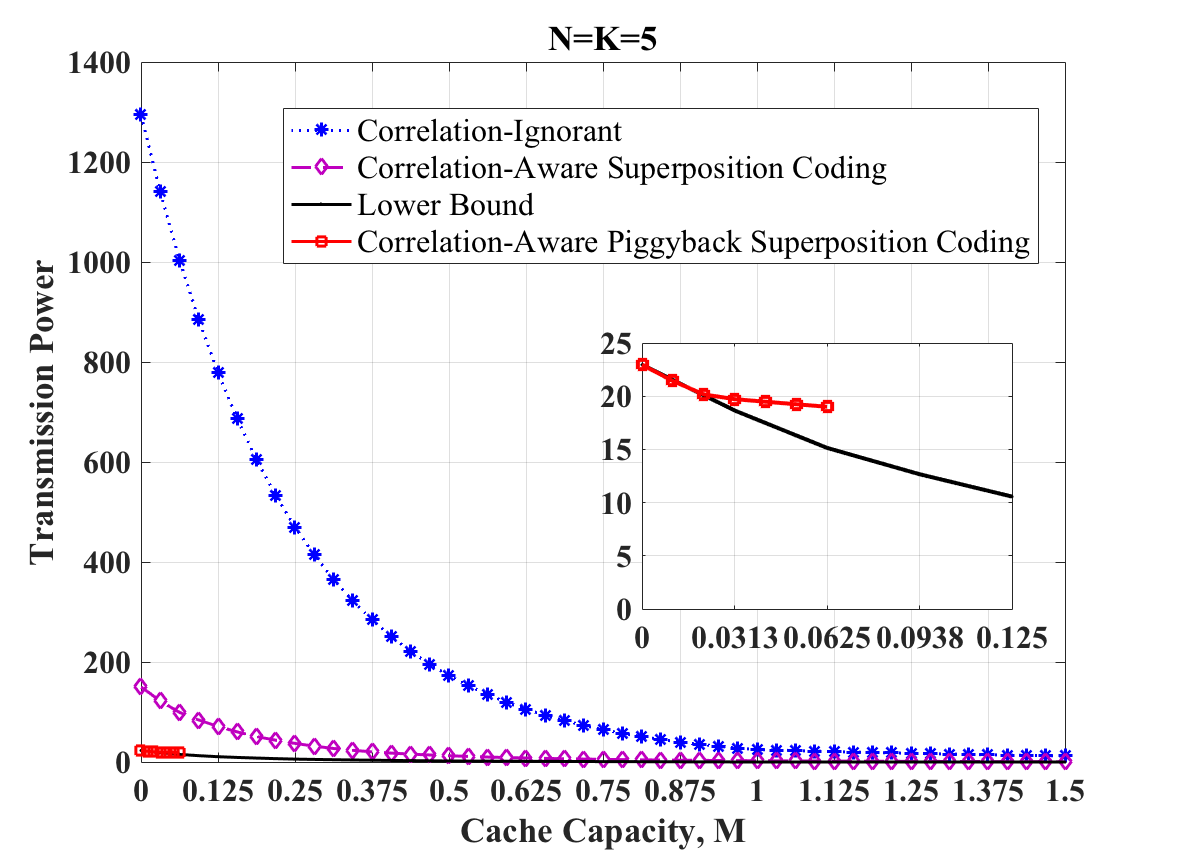}
\caption{Transmission power vs. cache capacity,$1/h_k^2=2-0.2(k-1)$, for $k=1, ..., K$. The portions of subfiles of different correlation level are specified by $\alpha_1=\alpha_5=1/16$, $\alpha_2=\alpha_4=1/4$, and $\alpha_3=3/8$. Correlation-aware superposition coding and piggyback superposition coding correspond to the schemes proposed in Section \ref{sec:scheme general} and Section \ref{subsec:piggy general}, respectively. }\label{fig:jointencoding1}
\end{figure}

\begin{figure}
\centering
\includegraphics[width=0.75\linewidth]{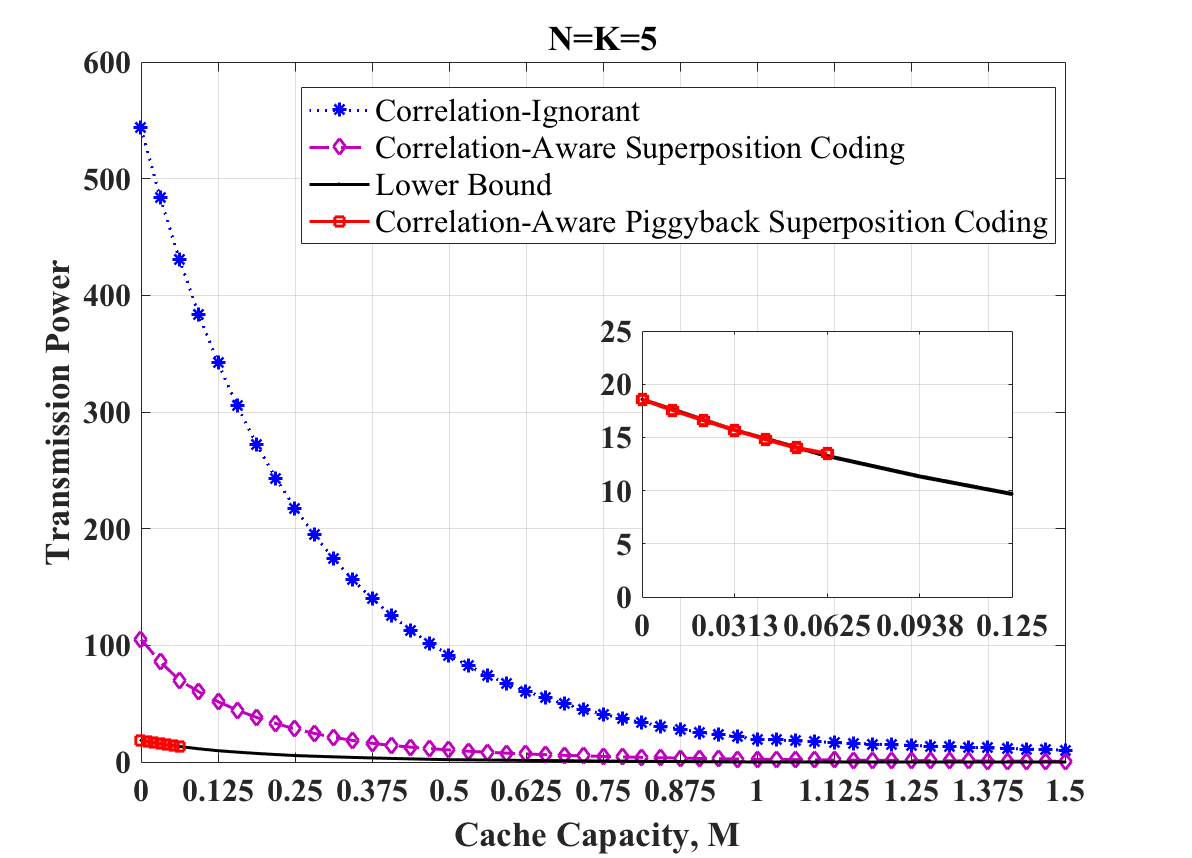}
\caption{Transmission power vs. cache capacity,$1/h_k^2=2-0.4(k-1)$, for $k=1, ..., K$. The portions of subfiles of different correlation level are specified by $\alpha_1=\alpha_5=1/16$, $\alpha_2=\alpha_4=1/4$, and $\alpha_3=3/8$. Correlation-aware superposition coding and piggyback superposition coding correspond to the schemes proposed in Section \ref{sec:scheme general} and Section \ref{subsec:piggy general}, respectively. }\label{fig:jointencoding2}
\end{figure}

Next, we consider the same setting with  $N=5$, $K=5$, and $R=1$ as before, but let $R_1=R_2=\cdots=R_N$, i.e., each subfile has the same size, which yields $\alpha_1=\alpha_5=1/16$, $\alpha_2=\alpha_4=1/4$, and $\alpha_3=3/8$. In Fig. \ref{fig:jointencoding1}, the channel gains are given as $1/h_k^2=2-0.2(k-1)$, for $k=1, ..., 5$, while in Fig. \ref{fig:jointencoding2},  $1/h_k^2=2-0.4(k-1)$, for $k=1, ..., 5$. We compare the proposed scheme presented in Section \ref{sec:scheme sep}, referred to as \textit{Correlation-Aware Superposition Coding}, the joint encoding scheme with coded placement presented in Section~\ref{sec:scheme joint}, referred to \textit{Correlation-Aware Piggyback Superpostion Coding}, with the correlation ignorant scheme, and the lower bound as well. In both cases, the joint encoding scheme with coded placement can be applied when $M\leq 1/16$. We observe that the correlation-aware schemes significantly outperform the correlation-ignorant scheme in terms of transmission power, and the joint encoding scheme with coded placement further improves the energy efficiency remarkably and achieves approximately the lower bound. However, while in the zoomed-in figure of Fig. \ref{fig:jointencoding2}, the joint scheme meets the lower bound, it can be seen in the zoomed-in figure of Fig. \ref{fig:jointencoding1} that the joint scheme results in a slightly higher transmission power than the lower bound when the cache capacity is larger than a certain value. That is because the channel of a stronger user is not good enough to receive all the contents (which are the cached contents at the weaker user) piggybacked on the message intended for the weaker user without any additional cost, such that \eqref{optimalcondition} is not satisfied.

\section{Conclusions}\label{sec:conc}
We have investigated caching and delivery of correlated contents over a $K$-user Gaussian BC for users with equal-capacity caches. Correlation among files is captured by the component subfiles shared among different subsets of files. We have first derived a lower bound on the minimum transmission power with which any possible demand combination can be satisfied, assuming uncoded cache placement. We have then presented two upper bounds on the memory-power trade-off with  correlation-aware cache-aided coding schemes. The first scheme generates coded packets according to user demands, which are then delivered to users using superposition coding, where each coded packet is targeted at the weakest user demanding it. We have also proposed a coded placement scheme with joint encoding, in which the cache contents and user demands are encoded jointly, such that the weak users can use their cache contents for decoding, while the stronger users can decode both without additional resources.

Our numerical results indicate that the proposed coding schemes greatly improve the energy-efficiency of delivery over Gaussian BCs compared to correlation-ignorant schemes. For small cache memory sizes, the joint encoding scheme with coded caching requires a lower transmit power, which meets the lower bound assuming uncoded placement. 
A tight lower bound without the limitation to uncoded placement is currently under investigation.  

\appendices
\section{Proof of Theorem 2}
To prove Theorem 2, we show the required transmission power by the proposed caching and delivery scheme presented in Section~\ref{sec:scheme sep} is upper bounded by $P_{UB}(M, \boldsymbol \pi)$ for any demand combination $\mathbf{d}$, given cache allocation vector $\boldsymbol \pi$. 

Recall that for a given demand combination $\mathbf{d}=(d_1, ..., d_K)$, $\mathcal{D}=\{d_1, ..., d_K\}$. For $\ell\in [N]$, $r\in [\max\{\ell-N+|\mathcal{D}|, 1\}: \min\{\ell, |\mathcal{D}|\}]$, ${\mathcal W}_r = \{\overline{W}_{\mathcal{S}}: |S| = \ell,\; |\mathcal{S}\cap\mathcal{D}|=r\}$ consists of $\binom{N-|\mathcal{D}|}{\ell-r}\binom{|\mathcal{D}|}{r}$ subfiles. Function GROUP generates $\binom{N-|\mathcal{D}|}{\ell-r}\binom{|\mathcal{D}|-1}{r-1}$ groups based on ${\mathcal W}_r$. For each group $\mathfrak{S}_i=(\mathcal{S}_1, ..., \mathcal{S}_K)$, Algorithm 1 runs function SINGLE-DEMAND twice (code line $6$ to $7$) to generates two set of coded messages $V_1^A$, ..., $V_K^A$ and $V_1^B$, ..., $V_K^B$ corresponding to $\{\overline{W}^A_{\mathcal S, \mathcal A}\}$ and $\{\overline{W}^B_{\mathcal S, \mathcal B}\}$, respectively. We recall that $\mathcal{K}$ is the set of the weakest users with distinct demands according to $\mathfrak{S}_i$, where $\mathcal{K}\triangleq\{k: \mathcal{S}_{k}\notin \{\mathcal{S}_{1}, ..., \mathcal{S}_{k-1}\}\}$ (line $4$ of function SINGLE-DEMAND), and denote by $e_k$ the number of leaders after user $k$, i.e., $e_k\triangleq \sum\limits_{k'=k+1}^K \mathbbm{1}\{k' \in \mathcal{K}\}$. Then for $k\in [K]$, the total size of $V_k^A$ and $V_k^B$ denoted by ${\hat{\gamma}}_{k,\ell, r}(\mathcal{K})$ (normalized by $n$), i.e., ${\hat{\gamma}}_{k,\ell, r}(\mathcal{K})\triangleq |V_k^A|+|V_k^B|$, is given by
\small
\begin{align}\label{gemma1}
{\hat{\gamma}}_{k,\ell, r}(\mathcal{K})=\begin{cases}
\frac{\binom{K-k}{t_\ell^A}}{\binom{K}{t_\ell^A}}(t_\ell^B-t_\ell)R_l+\frac{\binom{K-k}{t_\ell^B}}{\binom{K}{t_\ell^B}}(t_\ell-t_\ell^A)R_\ell~~~&\mbox{if}~~ k\in \mathcal{K},\\
\frac{\binom{K-k}{t_\ell^A}-\binom{K-k-e_k}{t_\ell^A}}{\binom{K}{t_\ell^A}}(t_\ell^B-t_\ell)R_\ell+\frac{\binom{K-k}{t_\ell^B}-\binom{K-k-e_k}{t_\ell^B}}{\binom{K}{t_\ell^B}}(t_\ell-t_\ell^A)R_\ell~~~&\mbox{if}~~ k\notin \mathcal{K}.
\end{cases}
\end{align}
\normalsize
Thus, the additional power required to send coded messages $V^A_1, ...., V_K^A$, and $V_1^B, ..., V_K^B$, denoted by $\Delta P$, is given as 
\small
\begin{align}
&\Delta P ({\hat{\gamma}}_{1,\ell, r}(\mathcal{K}), ..., {\hat{\gamma}}_{K,\ell, r}(\mathcal{K}))\nonumber\\
&\qquad\qquad=\sum\limits_{k=1}^{K}\left(\frac{2^{2(\overline{\rho}_k+{\hat{\gamma}}_{k,\ell, r}(\mathcal{K}))}-1}{h_k^2}\right) \prod\limits_{j=1}^{k-1}\frac{2^{2(\overline{\rho}_j+{\hat{\gamma}}_{j,\ell, r}(\mathcal{K}))}}{h_j^2}-\sum\limits_{k=1}^{K}\left(\frac{2^{2\overline{\rho}_k}-1}{h_k^2}\right) \prod\limits_{j=1}^{k-1}\frac{2^{2\overline{\rho}_j}}{h_j^2}, \notag
\end{align}
\normalsize
where $\overline{\rho}_1, ..., \overline{\rho}_1 \in \mathbbm{R}^+$ denote the total rate of all the other coded message required to be sent over the Gaussian BC.

Note that 
\small
\begin{align}\label{gemma2}
{\hat{\gamma}}_{k,\ell, r}([|\mathcal{K}|])=\begin{cases}
\frac{\binom{K-k}{t_\ell^A}}{\binom{K}{t_\ell^A}}(t_\ell^B-t_\ell)R_\ell+\frac{\binom{K-k}{t_\ell^B}}{\binom{K}{t_\ell^B}}(t_\ell-t_\ell^A)R_\ell~~~&\mbox{if}~~ k\in [|\mathcal{K}|],\\
0~~&\mbox{if}~~ k\notin [|\mathcal{K}|].
\end{cases}
\end{align}
\normalsize
Compare \eqref{gemma1} and \eqref{gemma2}. We have then
\small
\begin{equation}
\sum\limits_{k=1}^K {\hat{\gamma}}_{k,\ell, r}(\mathcal{K})=\sum\limits_{k=1}^K {\hat{\gamma}}_{k,\ell, r}([|\mathcal{K}|])=\frac{\binom{K}{t_\ell^A+1}-\binom{K-|\mathcal{K}|}{t_\ell^A+1}}{\binom{K}{t_\ell^A}}(t_\ell^B-t_\ell)R_\ell+\frac{\binom{K}{t_\ell^B+1}-\binom{K-|\mathcal{K}|}{t_\ell^B+1}}{\binom{K}{t_\ell^B}}(t_\ell-t_\ell^A)R_\ell,\notag     
\end{equation}
\normalsize
while ${\hat{\gamma}}_{k,\ell, r}([|\mathcal{K}|]) \geq {\hat{\gamma}}_{k,\ell, r}(\mathcal{K})$ if $k\in [E_d]$; ${\hat{\gamma}}_{k,\ell, r}([|\mathcal{K}|])\leq {\hat{\gamma}}_{k,\ell, r}(\mathcal{K})$ otherwise. It yields  
\small
\begin{equation}\label{plower}
\Delta P ({\hat{\gamma}}_{1,\ell, r}(\mathcal{K}), ..., {\hat{\gamma}}_{K,\ell, r}(\mathcal{K})) \leq \Delta P ({\hat{\gamma}}_{1,\ell, r}([|\mathcal{K}|]), ..., {\hat{\gamma}}_{K,\ell, r}([|\mathcal{K}|]))
\end{equation}
\normalsize
Note that each group generated by function GROUP has at most $\lceil |\mathcal{D}|/r \rceil+1$ distinct elements, which corresponds to at most $\lceil |\mathcal{D}|/r \rceil+1$ distinct elements by running function SINGLE-DEMAND, i.e., $|\mathcal{K}|\leq \lceil |\mathcal{D}|/r \rceil+1$. We have then ${\hat{\gamma}}_{k,\ell, r}([\lceil|\mathcal{D}|/r \rceil+1]) \geq {\hat{\gamma}}_{k,\ell, r}(\mathcal{K})$, $\forall k \in [K]$. With \eqref{plower},
\small
\begin{equation}\label{plower1}
\Delta P ({\hat{\gamma}}_{1,\ell, r}(\mathcal{K}), ..., {\hat{\gamma}}_{K,\ell, r}(\mathcal{K})) \leq \Delta P ({\hat{\gamma}}_{1,\ell, r}([\lceil|\mathcal{D}|/r \rceil+1]), ..., {\hat{\gamma}}_{K,\ell, r}([\lceil|\mathcal{D}|/r \rceil+1])).\nonumber
\end{equation}
\normalsize
Following the same procedure with all the groups, we can lower bound the total transmission power to satisfy demand combination $\mathbf{d}$ as follows
\small
\begin{align}
& P(M, \boldsymbol \pi, \mathcal{D})\leq \sum\limits_{k=1}^{ K}\left(\frac{2^{2\overline{\rho}_k}-1}{h_k^2}\right) \prod\limits_{j=1}^{k-1}{2^{2\overline{\rho}_j}}, \nonumber\\
&\overline{\rho}_k\triangleq \sum\limits_{\ell=1}^N \sum\limits_{r=\max\{\ell-N+|\mathcal{D}|, 1\}}^{\min\{\ell, |\mathcal{D}|\}  }\binom{N-|\mathcal{D}|}{\ell-r}\binom{|\mathcal{D}|}{r}\hat{\gamma}_{k, \ell, r}([\lceil|\mathcal{D}|/r \rceil+1]), \nonumber
\end{align}
\normalsize
which by letting $\mathcal{D}=[\min\{N, K\}]$, proves Theorem 2.

\section{Proof of Theorem~\ref{thm:piggy}}\label{app:piggy}
For a demand vector $\mathbf{d}$, $\forall \mathbf{d}\in [N]^K$, the proposed scheme presented in Section~\ref{subsec:piggy general} constructs a $Ne(\mathbf{d})$-level Gaussain superposition code. We denote the minimum total transmission power required by this scheme to satisfy $\mathbf{d}$ by $P(\mathbf{d}, M)=\sum\limits_{i=1}^{Ne(\mathbf{d})}P_i(\mathbf{d}, M)$, where $P_i(\mathbf{d}, M)$ is the power allocated to generate the $i^{\text{th}}$ level codeword. With \eqref{rowmessage} and \eqref{columnmessage}, we have 
\small
\begin{subequations}
\begin{align}
|V_{i,\mathbf{d}}^r|=M, |V_{i,\mathbf{d}}^c| =\rho_i-M, ~&\text{if}~k_i=i,\notag\\
|V_{i,\mathbf{d}}^r|=0, |V_{i,\mathbf{d}}^c| =\rho_i, ~&\text{if}~k_i\neq i.\notag
\end{align}
\end{subequations}
\normalsize
Thus, according to \eqref{eq:cond all} and \eqref{eq:cond row}, it yields 
\small
\begin{align}
P_i(\mathbf{d}, M)=\begin{cases}\max\bigg\{ \,  \Big(\frac{2^{2\tilde{\rho}_i} -1}{h^2_i}\Big)\Big(1+ h^2_i \, \sum\limits_{j=i+1}^{K}P_j \Big) ,\;\Big(\frac{2^{2(\tilde{\rho}_i+M)} -1}{h^2_{i+1}}\Big)\Big(1+ h^2_{i+1} \,
\sum\limits_{j=i+1}^{K}P_j \Big) \, \bigg\}, ~&\mbox{if}~k_i=i,\\
\Big(\frac{2^{2(\tilde{\rho}_i+M)} -1}{h^2_{i}}\Big)\Big(1+ h^2_{i} \,
\sum\limits_{j=i+1}^{K}P_j \Big), ~&\mbox{if}~k_i\neq i,
\end{cases}\notag
\end{align}
\normalsize
$\forall i\in [Ne(\mathbf{d})]$. It is straightforward to see that the worst-case demand combination $\mathbf{d}_{\text{worst}}$ that maximizes $P(\mathbf{d}, M)$, i.e., $\mathbf{d}_{\text{worst}}=\argmax\limits_{\mathbf{d}}P(\mathbf{d}, M)$, is such that $Ne(\mathbf{d})=\min\{N, K\}$ and $\mathcal U=[\min\{N, K\}]$, i.e., the weakest $\min\{N, K\}$ users request distinct files. And we have $P(\mathbf{d}_{\text{worst}}, M)=P^{\text{PB}}_{UB}(M)$, which completes the proof of Theorem~\ref{thm:piggy}.

\bibliographystyle{IEEEtran}
\bibliography{Journal}
\end{document}